\documentclass[11pt]{article}
\pdfoutput=1
\usepackage{amsmath,amssymb,amsfonts,amsxtra, mathrsfs, makeidx,graphics,graphicx,amsthm,epsfig, youngtab,jheppub1}

\newcommand{\be}{\begin{equation}}
\newcommand{\ee}{\end{equation}}
\newcommand{\beq}{\begin{equation}}
\newcommand{\eeq}{\end{equation}}
\newcommand{\ba}{\begin{array}}
\newcommand{\ea}{\end{array}}
\newcommand{\bi}{\begin{itemize}}
\newcommand{\ei}{\end{itemize}}
\def\bea#1\eea{\allowdisplaybreaks \begin{eqnarray}#1\end{eqnarray}}
\newcommand{\ben}{\begin{enumerate}}
\newcommand{\een}{\end{enumerate}}
\newcommand{\bean}{\begin{eqnarray*}}
\newcommand{\eean}{\end{eqnarray*}}
\newcommand{\eref}[1]{(\ref{#1})}

\newcommand{\tref}[1]{Table~\ref{#1}}

\newcommand{\nn}{\nonumber}

\newcommand{\tr}{\mathop{\rm Tr}}
\newcommand{\PE}{\mathop{\rm PE}}

\newcommand{\BZ}{\mathbb{Z}}

\newcommand{\BU}{\mathbf{1}}

\newcommand{\comment}[1]{}

\newcommand{\CD}{{\cal D}}

\newcommand{\CW}{{\cal W}}
\newcommand{\CN}{{\cal N}}

\newcommand{\CH}{{\cal H}}

\newcommand{\CC}{{\cal C}}

\newcommand{\CI}{{\cal I}}
\newcommand{\diag}{\mathrm{diag}}

\newcommand{\FC}{\mathfrak{C}}

\newcommand{\eg}{{\it e.g.}}

\newcommand{\ud}{\mathrm{d}}

\newcommand{\tH}{\widetilde{H}}
\newcommand{\tCH}{\widetilde{\mathcal{H}}}
\newcommand{\ta}{\widetilde{a}}

\newtheorem{theorem}{\bf Theorem}

\newtheorem{lemma}[theorem]{\bf Lemma}

\title{The Hilbert series of $\CN=1$ $SO(N_c)$ and $Sp(N_c)$ SQCD, Painlev\'e VI and Integrable Systems}

\author[a]{Estelle Basor,}
\author[b]{Yang Chen,}
\author[c]{and Noppadol Mekareeya} 

\affiliation[a]{American Institute of Mathematics, \\
360 Portage Avenue, Palo Alto, CA 9430, USA}
\affiliation[b]{Department of Mathematics, Imperial College London, \\
180 Queen's Gates, London SW7 2BZ, UK}
\affiliation[c]{Max-Planck-Institut f\"ur Physik (Werner-Heisenberg-Institut), \\
F\"ohringer Ring 6, 80805 M\"unchen, Deutschland}
\emailAdd{ebasor@aimath.org}
\emailAdd{ ychen@imperial.ac.uk}
\emailAdd{noppadol@mpp.mpg.de}
\abstract{We present a novel approach for computing the Hilbert series of 4d $\CN=1$ supersymmetric QCD with $SO(N_c)$ and $Sp(N_c)$ gauge groups.  It is shown that such Hilbert series can be recast in terms of determinants of Hankel matrices.  With the aid of results from random matrix theory, such Hankel determinants can be evaluated both exactly and asymptotically.  Several new results on Hilbert series for general numbers of colours and flavours are thus obtained in this paper.  We show that the Hilbert series give rise to families of rational solutions, with palindromic numerators, to the Painlev\'e VI equations.  Due to the presence of such Painlev\'e equations, there exist integrable Hamiltonian systems that describe the moduli spaces of $SO(N_c)$ and $Sp(N_c)$ SQCD.   To each system, we explicitly state the corresponding Hamiltonian and family of elliptic curves.  It turns out that such elliptic curves take the same form as the Seiberg--Witten curves for 4d $\CN=2$ $SU(2)$ gauge theory with 4 flavours.}

\begin{document}
\maketitle

\section{Introduction and summary}
Supersymmetric gauge theory has become one of the prime subjects of interest in quantum field theory and string theory. The presence of supersymmetry in such theories provides a better control of dynamics of the theory than non-supersymmetric ones.   Hence many aspects of the theories can be studied analytically and exactly.  Moreover, such theories exhibit a wide range of interesting phenomena, such as confinement, chiral symmetry breaking and dualities.

One of the simplest classes of supersymmetric gauge theories is Supersymmetric Quantum Chromodynamics (SQCD).  This class of theories has received a lot of attention, partly because of the richness of its quantum dynamics and partly because of its applications for dynamical supersymmetry breaking.  
In this paper, we focus on SQCD with $\CN=1$ supersymmetry in four dimensions and the gauge group being $SO(N_c)$ or $Sp(N_c)$.  The matter content consists of $N_f$ `flavours' of chiral superfields, called {\it quarks}, in the fundamental representation of the gauge group.  The global symmetries for the cases of $SO(N_c)$ gauge group and $Sp(N_c)$ gauge group are $U(N_f)$ and $U(2N_f)$ respectively.  Hence, in each case, the quarks transform in the bi-fundamental representation of $SO(N_c) \times U(N_f)$ or $Sp(N_c) \times U(2N_f)$ respectively.  The superpotential is taken to be zero.

Each of such theories has a continuous manifold of inequivalent exact ground states, known as a {\it moduli space of vacua} or simply a {\it moduli space}.  We refer the readers to \cite{Intriligator:1995id, hep-th/9505006} for detailed descriptions on the moduli spaces of $SO(N_c)$ and $Sp(N_c)$ SQCD.   In classical theory, the moduli space is the space of the solutions of the vacuum equations, namely the $F$ and $D$ term equations -- this is referred to as the {\it classical moduli space}.  This space can be viewed as an algebraic variety, generated by gauge invariant combinations of quarks (known as {\it mesons} and {\it baryons}).  These generators of the moduli space may be subject to certain relations among themselves.  

One of the aims of this paper is to compute a partition function, known as the {\it Hilbert series}, to count the number of gauge invariant quantities on the moduli space.  In fact, Hilbert series have been calculated and have been used to characterise moduli spaces of a wide range of supersymmetric gauge theories (see, \eg, \cite{Benvenuti:2006qr, Feng:2007ur, Benvenuti:2010pq, Hanany:2010qu, Hanany:2011db, Hanany:2008qc, Gray:2008yu, Hanany:2008kn, Hanany:2008sb, Chen:2011wn, JJK}). For each theory, the Hilbert series contains information about the generators and relations of the moduli space, viewed as an algebraic variety \cite{Benvenuti:2006qr, Feng:2007ur}. Moreover, given a Hilbert series, one can compute not only the dimension of the moduli space from the order of pole (see, \eg, \cite{Benvenuti:2006qr, Feng:2007ur} and \cite{DK}), but one can also use it as a test whether the moduli space is a Calabi-Yau variety \cite{Gray:2008yu,Hanany:2008kn}.

The Hilbert series of $SO(N_c)$ and $Sp(N_c)$ SQCD have been computed in \cite{Hanany:2008kn}.  Such Hilbert series were computed from the Molien--Weyl formula, which involve contour integrals over the torus $\mathbb{T}^r$, where $r$ is the rank of the gauge group.  In \cite{Hanany:2008kn}, such integrals were computed using the residue theorem for several simple cases.  From which, a number of general formulae were conjectured and checked using many non-trivial tests.  

In this paper, we introduce a new method in computing the Hilbert series of $SO(N_c)$ and $Sp(N_c)$ SQCD.  The idea of this method is similar to those of \cite{Chen:2011wn} and \cite{JJK}\footnote{We refer the readers to \cite{Balasubramanian:2004fz, Jokela:2005ha, Jokela:2008zh, Jokela:2009fd, Jokela:2010cc} for the use of Toeplitz determinants in the literature on decaying D-branes.}, where the Hilbert series of $U(N_c)$ and $SU(N_c)$ SQCD were recast in terms of determinants of Toeplitz matrices and these determinants were evaluated exactly and asymptotically using certain results from random matrix theory \cite{BottSil, BottcherWidom, BasorWidom, Geronimo:1979iy, ref:BO, ref:szego, ref:FisherHartwig1, ref:FisherHartwig2, ref:BottcherSilbermann}. On the other hand, in this paper we show that, for $SO(N_c)$ or $Sp(N_c)$ SQCD, the contour integrals in the Molien--Weyl formula can be rewritten in terms of determinant of a Hankel matrix.  For these theories, we find that it is possible to apply various versions of exact determinant formulae (EXDT), which are due to a result of \cite{BE2008} combined with a computation done in \cite{PJW1}, to calculate the Hankel determinants in the question both exactly and asymptotically.  With this method, the Hilbert series are computed for a class of values of $N_f$ and $N_c$ rather than a specific value of $(N_f,N_c)$ as in \cite{Hanany:2008kn}.  Such results for SQCD with the gauge groups $SO(2n+1)$, $SO(2n)$ and $Sp(n)$ are collected in Sections \ref{sec:Bn}, \ref{sec:Dn}, \ref{sec:Cn} (see also Appendix \ref{app:refined} for the refined Hilbert series).  Using this method of computations, we prove certain results in \cite{Hanany:2008kn} and obtain also several new results.

With the aid of \cite{DaiZhang}, we show that the Hankel determinants corresponding to $SO(N_c)$ and $Sp(N_c)$ SQCD with $N_f$ flavours give rise to infinite families of the Painlev\'e VI equation, with the parameters summarised in \tref{tab:parameters}.  These solutions share a common feature: When written in terms of an appropriate variable, they are rational functions with palindromic numerators.  We discuss such solutions and their properties in Section \ref{PainleveVI}.

In Section \ref{sec:integrability}, we make use of the Painlev\'e VI equations to explore further properties of the moduli spaces.  To each of such Painlev\'e equations, there exists a corresponding Hamiltonian system whose Hamiltonian describes the moduli space of $SO(N_c)$ and $Sp(N_c)$ SQCD with $N_f$ flavour.  The Hamiltonian is explicitly stated in \eref{Hamiltonian}.   

A result of \cite{Noumi} implies that there is a corresponding family of elliptic curves to such a Hamiltonian system.  These curves take the same form as the Seiberg--Witten curves of 4d $\CN=2$ $SU(2)$ gauge theory with 4 flavours \cite{Seiberg:1994aj}.  We discuss such elliptic curves and their parameters in Section \ref{sec:curves}.    Although these curves are naturally associated with the Painlev\'e VI equations, their physical origin and interpretation remain unclear.   Nevertheless, the presence of the Painlev\'e VI equation also implies the existence of a Lax pair which provides the integrability structure to the aforementioned Hamiltonian system.

Note however that the classical moduli space may be modified by quantum corrections.  
In Section \ref{sec:validity}, we use the results from \cite{Intriligator:1995id, hep-th/9505006} to briefly comment on the validity of our results for a quantum moduli space.  In the case that a quantum moduli space of supersymmetric vacua still exists, far away from singularities, the generators and their relations are unaffected by quantum effects and there are no extra massless degrees of freedom.  Hence the Hilbert series computed in the earlier sections still capture the structure of gauge invariant operators in such a region of the quantum moduli space.   We therefore conjecture that such a region of the quantum moduli space is still described by the integrable Hamiltonian system \eref{Hamiltonian}.\footnote{The correspondence between 4d $\CN=1$ gauge theories and  integrable systems has also been studied for a large class of models which admit brane tiling descriptions (also known as {\it dimer models}).  We refer the readers to \cite{GK, Franco:2011sz, Eager:2011dp}.}

\section{$SO(2n+1)$ SQCD with $N_f$ flavours} \label{sec:Bn}
In this section, we consider the Hilbert series of $SO(2n+1)$ SQCD with $N_f$ flavours.  We first write it in the form of the multi-contour integrals over the $n$-torus $\mathbb{T}^n$.  This form is then recast in terms of determinant of a Hankel matrix.   We then apply a version of the EXDT formula to compute such a determinant both exactly and asymptotically.

\subsection{The computations of Hilbert series}
The Hilbert series of $SO(N_c)$ SQCD with $N_f$ flavours can be computed in two steps as follows (see \cite{Gray:2008yu,Hanany:2008kn} for further details).

\paragraph{Step 1.}  We first consider the space of symmetric functions of the quarks $Q^i_a$ with $i =1, \ldots, N_f$ and $a = 1, \ldots, N_c$.  The Hilbert series of this space can be constructed using the {\it plethystic exponential}, which is a generator for symmetrisation \cite{Benvenuti:2006qr, Feng:2007ur}.  We define the plethystic exponential of a multi-variable function $g(t_1,...,t_n)$ that vanishes at the origin, $g(0,...,0) = 0$, to be
\bea
\PE [ g(t_1, \ldots, t_n) ] := \exp \left( \sum_{r=1}^\infty \frac{1}{r} g(t_1^r, \ldots, t_n^r) \right)~.
\eea
Note that the quarks transform in the vector representation $[1,0, \ldots, 0]$ of the $SO(N_c)$ gauge group and in the fundamental representation of $[1,0, \ldots, 0]$ of the $SU(N_f)$ flavour symmetry.  For reference, we write down the characters of these representations as follows:
\bea
[1,0, \ldots, 0]^{SO(2n+1)}_z &=& 1 + \sum_{a=1}^n \left( z_a + \frac{1}{z_a} \right)~, \qquad
\left[1,0, \ldots, 0 \right]^{SO(2n)}_z = \sum_{a=1}^n \left( z_a + \frac{1}{z_a} \right)~, \nn \\
\left[1,0, \ldots, 0 \right]^{SU(N_f)}_x 
&=& x_1 + \sum_{k=1}^{N_f -2} \frac{x_{k+1}}{x_k} + \frac{1}{x_{N_f-1}}~,
\eea
where, here and from now on, we use $z_a$ (with $a=1, \ldots, n$) to denote the $B_n = SO(2n+1)$ or $D_n = SO(2n)$ gauge fugacities and use $x_i$ (with $i =1, \ldots, N_f-1$) to denote the $A_{N_f-1}=SU(N_f)$ flavour fugacities. 

In this section, we focus on the $B_n = SO(2n+1)$ gauge group.  The Hilbert series for the space of symmetric functions of the quarks is then given by
\bea
\begin{array}{lll}
& \PE \big[ [1,0, \ldots,0]^{B_n}_z [1,0, \ldots,0]^{A_{N_f-1}}_x t \big]
  &=   \frac{1}{\left(1- t x_1 \right) \prod_{k=1}^{N_f-2} \left(1- t  \frac{x_{k+1}}{x_{k}} \right) \left(1- t x_{N_f-1} \right)}  \\
&  & \quad \times \prod_{a=1}^{n}  \frac{1}{\left(1- t z_a x_1 \right) \prod_{k=1}^{N_f-2} \left(1- t z_a \frac{x_{k+1}}{x_{k}} \right) \left(1- t z_a x_{N_f-1} \right)}  \\
& & \quad \times  \prod_{a=1}^{n}  \frac{1}{\left(1- t z_a^{-1} x_1 \right) \prod_{k=1}^{N_f-2} \left(1- t z_a^{-1} \frac{x_{k+1}}{x_{k}} \right) \left(1- t z_a^{-1} x_{N_f-1} \right)}  \\
&&  =   \PE \big[ [1,0, \ldots,0]_x t \big] \nn \\
&& \quad \times \prod_{a=1}^{n}  \frac{1}{\left(1- t z_a x_1 \right) \prod_{k=1}^{N_f-2} \left(1- t z_a \frac{x_{k+1}}{x_{k}} \right) \left(1- t z_a x_{N_f-1} \right)}  \\
&& \quad \times  \prod_{a=1}^{n}  \frac{1}{\left(1- t z_a^{-1} x_1 \right) \prod_{k=1}^{N_f-2} \left(1- t z_a^{-1} \frac{x_{k+1}}{x_{k}} \right) \left(1- t z_a^{-1} x_{N_f-1} \right)}~.
\end{array}
\eea
We may set the fugacities $x_1, \ldots, x_{N_f-1}$ to unity and obtain
\bea
 \frac{1}{(1-t)^{N_f} \prod_{a=1}^{n} (1-t z_a)^{N_f} (1-\frac{t}{z_a})^{N_f} }~.
\eea


\paragraph{Step 2.} Since the moduli space is parametrised by gauge invariant quantities, we need to project representations associated with symmetric functions in $Q$ discussed in Step 1 onto the trivial subrepresentation, which consists of the quantities invariant under the action of the gauge group.  Using knowledge from representation theory (known as the {\it Molien-Weyl formula} -- see \eg~\cite{DK, Djokovich}), this can be done by integrating over the whole gauge group.  

The Haar measures of the group $B_n = SO(2n+1)$ can be written in terms of contour integrations as
\bea
\int \ud \mu_{B_n} &=&  \frac{1}{(2 \pi i)^{n} n!} \oint \limits_{|z_1| =1}  \frac{\ud z_1}{z_1} \cdots \oint \limits_{|z_{n}| =1}  \frac{\ud z_{n}}{z_{n}} \left | \Delta_{n} \left(z + \frac{1}{z} \right) \right|^2 \prod_{a=1}^n \left[ 1- \frac{1}{2} \left(z_a + \frac{1}{z_a} \right)   \right]~. \nn \\
\eea

\paragraph{The refined Hilbert series.} The Hilbert series for $SO(N_c)$ SQCD with $N_f$ flavours is given by
\bea 
g_{N_f, B_n} (t ,x) &=& \int \ud \mu_{B_n} (z_1, \ldots, z_n)   \PE \left[ [1,0, \ldots,0]_x \left( 1 + \sum_{a=1}^n \left( z_a + \frac{1}{z_a} \right) \right) t \right]  \nn \\
&=&  \PE \left[ [1,0, \ldots,0]_x t \right] \CI_{N_f, B_n}(t,x)~, \label{HSBnPE}
\eea
where here and henceforth we write $x$ as collective notation for $x_1, \ldots, x_{n}$ and
\bea
\CI_{N_f, B_n}(t,x) &=& \int \ud \mu_{B_n} (z_1, \ldots, z_n) \PE \left[ [1,0, \ldots,0]_x \sum_{a=1}^n \left( z_a + \frac{1}{z_a} \right) t \right] \nn \\
&=& \frac{1}{(2 \pi i)^{n} n!} \oint \limits_{|z_1| =1}  \frac{\ud z_1}{z_1} \cdots \oint \limits_{|z_{n}| =1}  \frac{\ud z_{n}}{z_{n}} \left | \Delta_{n} \left(z + \frac{1}{z} \right) \right|^2 \prod_{a=1}^n \left[ 1- \frac{1}{2} \left(z_a + \frac{1}{z_a} \right)   \right] \nn \\
&& \times \PE \left[ [1,0, \ldots,0]_x \sum_{a=1}^n \left( z_a + \frac{1}{z_a} \right) t \right]~. \label{IBnref}
\eea

\paragraph{The unrefined Hilbert series.} Setting $x_1, \ldots, x_{n}$ to unity, we obtain the so-called {\it unrefined Hilbert series} of $SO(2n+1)$ SQCD with $N_f$ flavours:
\bea
g_{N_f, B_n} (t)&=&  \int \ud \mu_{B_n} (z_1, \ldots, z_n)  \frac{1}{(1-t)^{N_f} \prod_{a=1}^{n} (1-t z_a)^{N_f} (1-\frac{t}{z_a})^{N_f} } \nn \\
&=& \frac{1}{(1-t)^{N_f}}  \CI_{N_f, B_n}(t)~, \label{gNfBnintun}
\eea
where
\bea
 \CI_{N_f, B_n}(t) 
&=& \frac{1}{(2 \pi i)^{n} n!} \oint \limits_{|z_1| =1}  \frac{\ud z_1}{z_1} \cdots \oint \limits_{|z_{n}| =1}  \frac{\ud z_{n}}{z_{n}} \left | \Delta_{n} \left(z + \frac{1}{z} \right) \right|^2 \prod_{a=1}^n \frac{\left[ 1- \frac{1}{2} \left(z_a + \frac{1}{z_a} \right)   \right]}{(1-t z_a)^{N_f} (1-t/z_a)^{N_f}}~. \nn \\
\label{INfBnintun}
\eea

We focus on computations of the unrefined Hilbert series in the main text and postpone computations of the refined Hilbert series to Appendix \ref{app:refined}.

\subsection{The Hankel determinant}
In this section, we rewrite the integral form \eref{IBnref} and \eref{INfBnintun} of the Hilbert series in another way.  As we shall see below, the new form of the integrals allows us to recast \eref{IBnref} and \eref{INfBnintun} into a determinant known as a {\it Hankel determinant}.  It turns out that we can evaluate such a determinant in an exact way.  For simplicity, we shall postpone the discussion on the refined Hilbert series to Appendix \ref{AppBnref} and focus on the unrefined Hilbert series in this section. 

The computation of a Hankel determinant can be done using a generalization of an exact formula for Toeplitz determinants often referred to as the Borodin--Okounkov formula. This formula was actually first discovered by Geronimo and Case  \cite{Geronimo:1979iy}, later rediscovered by Borodin and Okounkov in relation to questions related to random matrix theory and then generalized to other classes of operators. The generalization found in \cite{BE2008} is the one that is useful for our purposes.  More will be said about this formula later in this section.
We now first rewrite the term \eref{INfBnintun} as a multiple integral over $[-1,1]$ and $[0,1]$.
 
\paragraph{Integrals over the intervals $[-1,1]^n$.} The integral $\CI_{N_f, B_n}(t)$ defined by \eref{INfBnintun} can be rewritten as
{\small
\bea
\CI_{N_f, B_n}(t) 
&=&  \frac{1}{(2 \pi i)^{n} n!} \oint \limits_{|z_1| =1}  \frac{\ud z_1}{z_1} \cdots \oint \limits_{|z_{n}| =1}  \frac{\ud z_{n}}{z_{n}}  \frac{\left | \Delta_{n} \left(z + \frac{1}{z} \right) \right|^2 \prod_{a=1}^n \left[ 1- \frac{1}{2} \left(z_a + \frac{1}{z_a} \right)   \right]}{\left[  \prod_{a =1}^{n} \left \{ 1 +t^2 - t \left(z_a + \frac{1}{z_a} \right) \right \} \right]^{N_f}} \nn \\
&=& \frac{2^{n^2-n}}{(2 \pi)^n n!} \int \limits_{-\pi}^\pi \ud \theta_1 \cdots \int \limits_{-\pi}^\pi \ud \theta_n \frac{ \prod_{1 \leq a< b \leq n}(\cos \theta_a - \cos \theta_b)^2 \prod_{a=1}^n (1- \cos \theta_a)}{ \prod_{a=1}^n (1- 2 t \cos \theta_a +t^2)^{N_f} }~. \nn \\
&=& \frac{2^{n^2}}{(2 \pi)^n n!} \int \limits_{0}^\pi \ud \theta_1 \cdots \int \limits_{0}^\pi \ud \theta_n \frac{ \prod_{1 \leq a < b \leq n}(\cos \theta_a - \cos \theta_b)^2 \prod_{a=1}^n (1- \cos \theta_a)}{ \prod_{a=1}^n (1- 2 t \cos \theta_a +t^2)^{N_f} }~. \nn \\
&=&\frac{2^{n^2}}{(2 \pi)^n n!} \int \limits_{-1}^1 \ud y_1 \cdots \int \limits_{-1}^1 \ud y_n  \prod_{1 \leq a<b \leq n} (y_a - y_b)^2 \prod_{a=1}^n  \frac{(1- y_a)^{1/2}(1+y_a)^{-1/2}}{ (1- 2 t y_a +t^2)^{N_f} } \label{intBnhankel}~, \qquad \qquad
\eea}
where we have written $z_a = e^{i \theta_a}$ and taken $y_a = \cos \theta_a$.

\paragraph{Integrals over the intervals $[0,1]^n$.} We can further rewrite the Hilbert series in terms of the Hankel determinants as follows.  Let us change the variable
\bea
y_a = 2 \zeta_a-1~. 
\eea
Therefore, we have
{\small
\bea
g_{N_f, B_n} (t) 
&=& \frac{1}{(1-t)^{N_f}}  \CI_{N_f, B_n}(t) \nn \\
&=& \frac{2^{2 n^2}}{(2 \pi)^n n!} \frac{1}{(1-t)^{N_f} (-4t)^{n N_f}} \int \limits_{0}^1 \ud \zeta_1 \cdots \int \limits_{0}^1 \ud \zeta_n \prod_{1 \leq a < b \leq n} (\zeta_a - \zeta_b)^2 \prod_{a=1}^n  \frac{\zeta_a^{-1/2} (1- \zeta_a)^{1/2}}{ \left(\zeta_a- \frac{(1+t)^2}{4t} \right)^{N_f} } \nn \\
&=& C_{N_f,B_n}(t) D_{N_f,B_n} (T)~,    \label{HShankelBn}
\eea}
where the factor $C_{N_f,B_n}(t)$ is given by
\bea
C_{N_f,B_n}(t) :=\frac{2^{2 n^2}}{(2 \pi)^n} \frac{1}{(1-t)^{N_f} (-4 t)^{n N_f}}~,
\eea
and $D_{N_f,B_n} (T)$ and the variable $T$ are given by
\bea
D_{N_f,B_n} (T) &:=& \frac{1}{n!} \int \limits_{0}^1 \ud \zeta_1 \cdots \int \limits_{0}^1 \ud \zeta_n \prod_{1 \leq a < b \leq n} (\zeta_a - \zeta_b)^2 \prod_{a=1}^n w(\zeta_a;T)~, \label{Hankelsoodd} \\
w(\zeta_a;T) &:=&   \zeta_a^{-1/2} (1- \zeta_a)^{1/2}(\zeta_a - T)^{-N_f}~,  \\
T &:=& \frac{(1+t)^2}{4t}~.
\eea
Using Gram's formula (see \eg, \cite{Mehta} and Appendix A of \cite{Chen:2011wn}), we rewrite $D_{N_f,B_n} (T)$ as determinant of the $n \times n$ matrix:
\bea
D_{N_f,B_n} (T) = \det \left( \int_0^1 \ud \zeta~w(\zeta;T) \zeta^{i+j} \right)_{i,j =0}^{n-1}~.
\eea

\paragraph{The Hankel determinant.}  The determinant $D_{N_f,B_n} (T)$ is known as the {\it Hankel determinant} with the perturbed Jacobi weight
\bea
w(\zeta;T) &=&   \zeta^{\alpha} (1- \zeta)^{\beta}(\zeta - T)^{\gamma}~,
\eea
with the parameters
\bea
\alpha = -1/2, \qquad \beta= 1/2, \qquad \gamma = - N_f~. \label{abcBn}
\eea

The weight $w(\zeta;T)$ is a perturbation $(\zeta - T)^{\gamma}$ on the Jacobi weight $\zeta^{\alpha} (1- \zeta)^{\beta}$.  Note that perturbed Jacobi weights have been extensively studied in \cite{PJW1,PJW2,PJW3,PJW4,PJW5,PJW6,PJW7,PJW8,PJW9} using the ladder operator approach to orthogonal polynomials.

\paragraph{Palindromic numerator of the Hilbert series.}  Observe that $T$ is invariant under $t \mapsto 1/t$.  Hence it is clear that the determinant $D_{N_f,B_n} (T)$ is also invariant under this transformation.  Since 
\bea
C_{N_f, B_n} (1/t) = (-1)^{N_f} t^{(2n+1) N_f} C_{N_f B_n} (t)~,
\eea
it follows that
\bea
g_{N_f, B_n} (1/t) =  (-t)^{(2n+1) N_f} g_{N_f, B_n} (t)=  (-t)^{N_c N_f} g_{N_f, B_n} (t)~, \label{invertBn}
\eea
and so the numerator of $g_{N_f, B_n}(t)$ is palindromic.  Note that this is the same argument used in \cite{Gray:2008yu} to prove the palindromic property of the numerator of $SU(N_c)$ SQCD.  A physical implication is that the moduli space is a Calabi-Yau variety.  

We shall make use of the transformation $t \mapsto 1/t$ to deduce some symmetries of certain Painlev\'e solutions in Section \ref{sec:solutionsP6}.

\subsection{The EXDT II formula}
Having identified the Hilbert series with the Hankel determinant, we can now apply techniques from random matrix theory to compute the Hankel determinant both exactly and asymptotically in a similar fashion to \cite{Chen:2011wn}.  As mentioned earlier, this technique involves a generalisation of the Geronimo--Case--Borodin--Okounkov (GCBO), which was used extensively in \cite{Chen:2011wn} for computing Toeplitz determinants.  Such a generalisation to the Hankel determinants with Jacobi weights and with the parameters $\alpha, \beta = \pm 1/2$  follows from a result  due to Basor--Ehrhardt \cite{BE2008} combined with a computation done in \cite{PJW1}.  Because of the choice of parameters there are four cases of the formula and we will denote these as EXDT I, EXDT II, etc. or EXDT when the choice of parameters is not specified. (The choice of I, II, etc. follows the notation in \cite{BE2008}.)

\subsubsection{Computation of the Hankel determinant}
To compute the integral \eref{intBnhankel} we observe that the factor $(1- y_a)^{1/2}(1+y_a)^{-1/2}$ in \eref{intBnhankel} tells us that the Hankel determinant \eref{intBnhankel} corresponds to Case II of Propositions 3.1 and 3.3 of \cite{BE2008} (see also Page 16 of \cite{PJW2}).  Subsequently, we apply Proposition 4.1 of \cite{BE2008} (EXDT II formula) to compute exact expressions of Hilbert series. 

In order to apply the formula we must define the several quantities.

\paragraph{The symbol and its factorisation.} The symbol for our problem is
\bea
a(z) :=  (1-t z)^{-N_f} (1-t/z)^{-N_f}~.  \label{symbBn}
\eea
Put $a(z) = a_+(z) \ta_+(z)$ with
\bea
a_+(z) =  (1-tz)^{-N_f}~, \qquad \ta_+(z) = (1-t/z)^{-N_f}~.
\eea
Define the function $c(z)$ as
\bea
c(z) = a_+^{-1} (z) \ta_+ (z) = \left( \frac{1- tz}{1-t/z} \right)^{N_f}~.
\eea
\paragraph{Fourier coefficients and related matrices.} The Fourier coefficients $c_k$ (with $k \in \BZ$) of a function $c(z)$ are defined by
\bea
c_k := \frac{1}{2 \pi i} \oint_{|z|=1} \frac{\ud z}{z} z^{-k} c(z) = \frac{1}{2 \pi i} \oint_{|z|=1} \frac{\ud z}{z} z^{-k}  \left( \frac{1- tz}{1-t/z} \right)^{N_f}~.
\eea
Note that the coefficients $c_k$ have been computed explicitly in (2.53) of \cite{Chen:2011wn}:
\bea
c_k &=&  \sum_{m=0}^{N_f-k} {N_f \choose k+m}  {N_f +m-1 \choose m} (-1)^{m+k} t^{2m+k} \nn \\
&=& (-t)^k {N_f \choose k} {}_2F_1 (k - N_f, N_f; k+1; t^2)~. \label{ck2F1}
\eea

\paragraph{The matrices $K^B$ and $K^B_n$.} From Case II on Page 13 of \cite{BE2008}, define an infinite matrix $K^B$ to be the such that the $(i, j)$-entry (with $i, j =0,1,2,\ldots$) is given by
\bea
K^B(i,j) = - c_{i+j +1}~, \label{KijdefBn}
\eea
where the superscript $B$ indicates the gauge group $B_n = SO(2n+1)$.

Let us define the projection matrix $Q_n$ to be an infinite matrix such that
\bea
Q_n = \diag(\underbrace{0, 0, \ldots,0}_{n~\text{zeros}}, 1, 1, \ldots)~, \label{Qndef}
\eea
and define the matrix $K^B_n$ to be
\bea
K^B_n = Q_n K^B Q_n~.
\eea
It follows that
\bea  \label{zeroKn} 
K^B_n(i, j) =\left\{ 
  \begin{array}{l l}
0 &\quad \text{for}~ 0 \leq i, j \leq n-1~\text{and}~ i+j \geq N_f\\
-c_{i+j+1} &\quad \text{otherwise}~.
\end{array} \right. 
\eea

\paragraph{EXDT II formula.} We put this all together now to compute  \eref{intBnhankel}  from the EXDT II formula (for more details see Proposition 4.1 of \cite{BE2008}):
\bea
\CI_{N_f, B_n}(t) = G(a)^n \widehat{F}_{II} (a) \det (\BU + K^B_n) ~,
\eea
where the function $G(a)$ is defined by
\bea
G(a) := \exp \left(\frac{1}{2 \pi} \oint_{|z|=1} \frac{\ud z}{z} \log a(z) \right) = 1~,
\eea
and the function $\widehat{F}_{II} (a)$ is given by (see Proposition 3.3 of \cite{BE2008}):
\bea
\widehat{F}_{II} (a) &=& \exp \left(- \sum_{n=0}^\infty [\log a]_{2n+1} +\frac{1}{2} \sum_{n=1}^\infty n [\log a]_n^2 \right) \nn \\
&=&  \exp \left(- N_f  \sum_{n=0}^\infty \frac{t^{2n+1}}{2n+1} +  \frac{1}{2}N_f^2 \sum_{n=1}^\infty n \times \frac{t^{2n}}{n^2} \right) \nn \\
&=& \exp \left( \frac{1}{2}N_f \log \left(\frac{1-t}{1+t} \right) - \frac{1}{2} N_f^2 \log(1-t^2) \right) \nn \\
&=& (1-t)^{N_f} (1-t^2)^{-N_f(N_f +1)/2}~.
\eea

\paragraph{The Hilbert series.} The Hilbert series is then given by
\bea
g_{N_f, B_n} (t) &=& (1- t)^{-N_f}\CI_{N_f, B_n}(t) \nn \\
&=&  (1- t)^{-N_f} G(a)^n \widehat{F}_{II} (a) \det (\BU + K^B_n) \nn \\
&=&  \frac{\det (\BU + K^B_n)}{(1-t^2)^{N_f(N_f +1)/2}}~. \label{unrefBOBEHS}
\eea 

\subsubsection{Some explicit examples}
In this subsection, we derive from \eref{unrefBOBEHS} some explicit expressions for the Hilbert series.

\paragraph{The case of $N_f < 2n+1$.}  In this case $K^B_n (i,j ) =0$ for all $i, j$.  Therefore, the Hilbert series is
\bea
g_{N_f < 2n+1} (t) = \frac{1}{(1-t^2)^{N_f(N_f+1)/2}}~. \label{FGBn}
\eea

\paragraph{The case of $N_f = 2n+1$.}  In this case 
\bea
K^B_n (i,j ) = \left\{ 
  \begin{array}{l l}
    t^{N_f} & \quad \text{if $i = j =n$}\\
    0 & \quad \text{otherwise}~.
  \end{array} \right. 
\eea
The Hilbert series is thus
\bea
g_{N_f = 2n+1, B_n} (t) = \frac{1+t^{N_f}}{(1-t^2)^{N_f(N_f+1)/2}}~. \label{CIBn}
\eea

\paragraph{The case of $N_f = 2n+2$.}  The non-trivial block of the matrix $K^B_n$ is given by
\bea
\begin{pmatrix}
 2 (1+n) t^{1+2 n} \left(1-t^2\right) \qquad& -t^{2+2 n} \\
 -t^{2+2 n} & 0
\end{pmatrix}~.
\eea
Therefore the Hilbert series is
\bea
g_{N_f = 2n+2, B_n} (t) &=& \frac{1+2 (1+n) t^{1+2 n}-2 (1+n) t^{3+2 n}-t^{4+4 n}}{ (1-t^2)^{N_f(N_f+1)/2} }~. \label{2np2Bn}
\eea

\paragraph{The case of $N_f = 2n+3$.}  The non-trivial block of the matrix $K^B_n$ is given by
{\small
\bea
\left(
\begin{array}{ccc}
 (3+2 n) t^{1+2 n} \left(1-t^2\right) \left[1-2 t^2+n \left(1-t^2\right)\right] \quad & (3+2 n) t^{2+2 n} \left(-1+t^2\right) \quad & t^{3+2 n} \\
 (3+2 n) t^{2+2 n} \left(-1+t^2\right) & t^{3+2 n} & 0 \\
 t^{3+2 n} & 0 & 0
\end{array}
\right)~. \nn \\
\eea}
Therefore the Hilbert series is
\bea
g_{N_f = 2n+3, B_n} (t) &=& \frac{1}{(1-t^2)^{N_f(N_f+1)/2}} \Big[ 1+(1+n) (3+2 n) t^{1+2 n}-4 (1+n) (2+n) t^{3+2 n} \nn \\
&& +(2 + n) (3 + 2 n) t^{5+2 n}-(2 + n) (3 + 2 n) t^{4+4 n}+4 (1+n) (2+n)t^{6+4 n} \nn \\
&& -(1+n) (3+2 n) t^{8+4 n}-t^{9+6 n} \Big]~. \label{2np3Bn}
\eea

\subsubsection*{Comments on the results}  
Let us state some comments on the above results.

\ben
\item The results for larger $\Delta = N_f -(2n+1)$ can also been obtained in a straightforward way. 
However, since such results are too long to be reported here, we only show the explicit results up to only $\Delta=2$.  The asymptotic formula for $n, N_f \rightarrow \infty$, with $\Delta$ held fixed and being finite, is derived in the next subsection.

\item Equations \eref{FGBn} and \eref{CIBn} were obtained in \cite{Hanany:2008kn} based on physical arguments that when $N_f < 2n+1$ the moduli space is {\it freely generated} by the mesons, and when $N_f = 2n+1$ the moduli space is a {\it complete intersection}, whose generators are mesons and a baryon subject to precisely one relation.  In \cite{Hanany:2008kn}, such equations were also confirmed by a few examples which were directly computable from \eref{INfBnintun} using the residue theorem.  In this paper, we not only prove such equations using the EXDT II formula, but new results, such as \eref{2np2Bn} and \eref{2np3Bn}, are also computed.  The latter are very difficult to be obtained by performing direct integrations or even by re-summing the character expansion (2.29) of \cite{Hanany:2008kn}.  (Such a chacteracter expansion is also re-stated in \eref{genchaexpBn}.)

\item It follows from Case II of Proposition 3.1 and Proposition 4.1 of \cite{BE2008} that the Hilbert series can be written in terms of determinant of the Toeplitz matrix minus the Hankel matrix with the symbol $a(z)$ given by \eref{symbBn}.  It is then clear that the Hilbert series is a rational function in $t$.  Furthermore, by considering the transformation property of $a(z)$ under $t \mapsto 1/t$, it is immediate that the numerator of the Hilbert series is palindromic.
\een

\subsubsection{Asymptotics of unrefined Hilbert series}
Having been computing several examples using \eref{unrefBOBEHS}, we now examine leading behaviour of the numerator $\det (\BU + K^B_n)$ of the Hilbert series as $n \rightarrow \infty$. As we shall see below, this leads to asymptotic formulae for unrefined Hilbert series in various limits.  For convenience, let us define
\bea
\Delta := N_f - N_c= N_f - (2n+1)~.
\eea
Note that for $\Delta \leq 0$, the exact result are given by \eref{FGBn} and \eref{CIBn}.  In this subsection, we shall henceforth assume that $\Delta > 0$. 

\subsubsection*{Leading behaviour of numerators of unrefined Hilbert series}

Consider
\bea \label{detformula}
\det (\BU + K^B_n) &=& \sum_{r=0}^\infty \frac{1}{r!} \left[ \sum_{s=1}^\infty \frac{1}{s} (-1)^{s+1} \tr ({K^B_n}^s) \right]^r \nn \\
&=&  1+ \tr K^B_n + \frac{1}{2} \left[ (\tr K^B_n)^2 - \tr ({K^B_n}^2) \right] - \frac{1}{2} (\tr K^B_n)\tr ({K^B_n}^2) + \ldots~. \nn \\
\eea
Using \eref{ck2F1} and \eref{zeroKn}, we see that the smallest order term in $1+ \tr K^B_n$ is $O(t^{2n+1+2\Delta})$.  For the terms $\left[ (\tr K^B_n)^2 - \tr ({K^B_n}^2) \right]$, it can be checked, rather delicately, that the highest contribution is $O(t^{4n+4})$.  Note that the terms with higher traces are smaller.  

In order that there are no overlaps between $1+ \tr K^B_n$ and $\left[ (\tr K^B_n)^2 - \tr ({K^B_n}^2) \right]$, let us {\it assume} that $\Delta$ is $O(1)$ as $n \rightarrow \infty$.  Then, the coefficients of $t^{2n+1+2k}$, for all $k \leq \Delta$, can be extracted from $1+\tr K^B_n$ and the subleading terms can be neglected.  Another way to see this is the following. The determinant in question is really a finite determinant of size $\Delta +1$ by $\Delta +1$ with $n$ sufficiently large. Consider all the products the determinant computation. If not all diagonal elements are in the product there must be at least two off-diagonal terms, and thus something of order greater that $t^{4+4n}$. If we use all the diagonals, then combining any two terms that involve powers of $t$ greater than zero will lead to terms of order at least $t^{4+4n}.$



Recall from \eref{zeroKn} that $K^B_n(i, j) =0 \quad \text{for all}~ 1 \leq i, j \leq n-1~\text{and}~ i+j \geq N_f$.
Therefore, we find that
\bea
\tr K^B_n &=& \sum_{l=0}^{\Delta} K^B_n (n+l,n+l)  \nn \\
&=& - \sum_{l=0}^{\Delta} c_{2n+1+2l} \nn \\
&=& - \sum_{l=0}^{\Delta} \sum_{j = 0}^{\Delta-2l}  {N_{f} \choose  2n+1+2l+ j}  {N_{f}+ j-1 \choose j}  (-1)^{j+2n+1+2l} t^{2j +2n+1+2l} \nn \\
&=& \sum_{l=0}^{\Delta} \sum_{j = 0}^{\Delta-2l}  {2n+1+\Delta \choose  \Delta-2l- j}  {2n+\Delta + j \choose j}  (-1)^{j} t^{2j +2n+1+2l}~.
\eea
Now let us extract the coefficient of $t^{2n+1+2k}$ (with $k \leq n$). This can be obtained when $j =  k-l$ (with the constraint $0 \leq j = k-l  \leq \Delta-2l$).  Therefore, such a coefficient can be written as
\bea
\CC_{2n+1+2k} := \sum _{l = 0}^{\min(k,\Delta-k)} {2n + 1+ \Delta \choose \Delta-k-l }{2n +\Delta +k-l \choose k-l}(-1)^{k - l}~. 
\eea

In order to compute this summation, we use the following lemma:  

\begin{lemma}  \label{estellesum} The following alternating sum can be evaluated as follows:
\bea
&& \sum_{l=0}^j (-1)^l { N \choose j - l }{N + h - l \choose h+1 - l } \nn\\
&=& { N \choose j}{N + h \choose h+1} - { N \choose j -1}{N +h -1 \choose h}  + \ldots + (-1)^j { N \choose 0}{N+h -j \choose h -j +1} \nn \\
&=& { N + h+1 \choose j}{N+h -j \choose h+1}~. \label{sumsum}
\eea
%
\end{lemma}

\begin{proof}
Let us define $F(N, j, h)$ to be the sum on the left hand side. Using the property that 
\bea 
{N \choose k} = {N-1 \choose k} + {N-1 \choose k-1}~, \label{binomprop}
\eea
we can split the alternating sum into three terms and have the recurrence relation:
\bea 
F(N, j, h)  = F(N - 1, j, h) + F(N-1, j -1, h) +F(N, j, h-1)~. \label{recurrence}
\eea

Now fix $j \leq N$ and do an induction on the quantity $N + h$. It is convenient to consider the pairs $(N, h)$ on an integer lattice.  Let us focus on the line $N+h=c$, where $c$ is a constant.  Suppose that every point under this line satisfies the lemma.  In order to prove the statement for the line $N+h =c+1$, we simply go one step to the left and then one step down according to the identity \eref{recurrence}.  At the point $(0, 0)$, this is clearly true so the lemma holds.  We need to also check that the formula works on the ``edges'' of the lattice where $N$ or $h$ are zero, but this is clearly true. Finally, it is easy to check that the right hand side of \eref{recurrence} adds up to \eref{sumsum}. This yields the result.
\end{proof}

Using Lemma \eref{estellesum} by setting $N = 2n+1+\Delta,~ j = \Delta-k,~ h= k-1$, we find that for $k \leq \Delta$ the coefficient of $t^{2n+1+2k}$ is
\bea 
\CC_{2n+1+2k} = (-1)^{k}{2n + 1+\Delta +k \choose \Delta -k}{ 2n +2k \choose k}~.
\eea
Thus,  the leading behaviour of the numerator $\det(\BU + K^B_n)$ in the limit $n \rightarrow \infty$ is
\bea
\det (\BU+K^B_n ) &\sim& 1+ \sum_{k=0}^\Delta \CC_{2n+1+2k} t^{2n+1+2k} + O(t^{4n+4}) \nn \\
&=& 1+ \sum_{k=0}^\Delta (-1)^k {2n + 1+\Delta +k \choose \Delta -k}{ 2n +2k \choose k} t^{2n+1+2k} + O(t^{4n+4}) \nn \\
&=& 1+ \sum_{k=0}^\Delta (-1)^k {N_f+k \choose \Delta -k}{ 2n +2k \choose k} t^{2n+1+2k} + O(t^{4n+4}) \nn \\
&=& 1 + t^{N_c} {N_f \choose \Delta} {}_3F_2 \left( \frac{1}{2}N_c,-\Delta ,N_f+1; \frac{1}{2}N_c +1,N_c; t^2 \right) \nn \\
&& + O(t^{2N_c+2})~. \label{asympBn}
\eea
It follows that the asymptotic formula (as $n \rightarrow \infty$) for the Hilbert series of $SO(N_c = 2n+1)$ SQCD with $N_f$ flavours and $\Delta = N_f - N_c = O(1)$ is
\bea
g_{N_f, SO(N_c)} (t) &\sim& \frac{1}{(1-t^2)^{N_f(N_f+1)/2} } \Big[ 1 + t^{N_c} {N_f \choose \Delta}  \times \nn \\
&& {}_3F_2 \left( \frac{1}{2}N_c,-\Delta ,N_f+1; \frac{1}{2}N_c +1,N_c; t^2 \right) + O(t^{2N_c+2}) \Big]~. \label{asympHSBn}
\eea

\paragraph{Example: $N_f = 2n+3$.} We have
\bea
\det (\BU+K^B_n ) &\sim&  1+(1+n) (3+2 n) t^{1+2 n}-4 (1+n) (2+n) t^{3+2 n} \nn \\
&& +(2 + n) (3 + 2 n) t^{5+2 n} + O(t^{4n+4})~.
\eea
This is in agreement with the `first half' of the numerator of \eref{2np3Bn}.

\section{$SO(2n)$ SQCD with $N_f$ flavours} \label{sec:Dn}
In this section, we examine the Hilbert series of SQCD with $D_n = SO(2n)$ gauge group and $N_f$ flavours of quarks.  The analogous calculations that were done in $B_n = SO(2n+1)$ case are done here.  Subsequently, we will see below that such Hilbert series can also be recast in terms of a Hankel determinant but with different parameters from those of $B_n$ case.  Hence we can apply a version of the EXDT formula to compute the determinant in a similar way as before.  The only essential difference is the use of EXDT IV instead of the II case employed earlier. We supply the details for completeness sake.

\subsection{The computations of Hilbert series}
The Haar measures of the group $D_n = SO(2n)$ is given by
\bea
\int \ud \mu_{D_n} &=&  \frac{2^{-(n-1)}}{(2 \pi i)^{n} n!} \oint \limits_{|z_1| =1}  \frac{\ud z_1}{z_1} \cdots \oint \limits_{|z_{n}| =1}  \frac{\ud z_{n}}{z_{n}} \left | \Delta_{n} \left(z + \frac{1}{z} \right) \right|^2~.
\eea
The refined Hilbert series for $SO(2n)$ SQCD with $N_f$ flavours can be written as
\bea
g_{N_f, D_n} (t ,x) &=& \int \ud \mu_{D_n} (z_1, \ldots, z_n)   \PE \left[ [1,0, \ldots,0]_x \sum_{a=1}^n \left( z_a + \frac{1}{z_a} \right) t \right] \nn \\
&=&  \frac{2^{-(n-1)}}{(2 \pi i)^{n} n!} \oint \limits_{|z_1| =1}  \frac{\ud z_1}{z_1} \cdots \oint \limits_{|z_{n}| =1}  \frac{\ud z_{n}}{z_{n}} \left | \Delta_{n} \left(z + \frac{1}{z} \right) \right|^2  \nn \\
&& \qquad \qquad \qquad \times \PE \left[ [1,0, \ldots,0]_x \sum_{a=1}^n \left( z_a + \frac{1}{z_a} \right) t \right] ~. \label{inthsrefnfdn}
\eea
Setting $x_1, \ldots, x_n =1$, we obtain the unrefined Hilbert series
\bea
g_{N_f, D_n}(t) &=&  \frac{2^{-(n-1)}}{(2 \pi i)^{n} n!} \oint \limits_{|z_1| =1}  \frac{\ud z_1}{z_1} \cdots \oint \limits_{|z_{n}| =1}  \frac{\ud z_{n}}{z_{n}} \left | \Delta_{n} \left(z + \frac{1}{z} \right) \right|^2   \prod_{a=1}^n \frac{1}{(1-t z_a)^{N_f} (1-t/z_a)^{N_f}}~. 
\nn \\ \label{gNfDnunref}
\eea

\subsection{The Hankel determinant}
In this section, we rewrite the multi-complex-contour integrals \eref{gNfDnunref} in terms of integrals over the intervals $[-1,1]$ and $[0,1]$. The latter form of the integrals allow us to recast the Hilbert series in terms of the Hankel determinant.

\paragraph{Integrals over the intervals $[-1,1]^n$.} Writing $z_a = e^{i \theta_a}$ and  $y_a = \cos \theta_a$ in \eref{gNfDnunref}, we obtain
{\small
\bea
g_{N_f, D_n}(t) 
&=&\frac{2^{(n-1)^2+n}}{(2 \pi)^n n!}  \int \limits_{-1}^1 \ud y_1 \cdots \int \limits_{-1}^1 \ud y_n  \prod_{1 \leq a<b \leq n} (y_a - y_b)^2 \prod_{a=1}^n  \frac{(1- y_a)^{-1/2}(1+y_a)^{-1/2}}{ (1- 2 t y_a +t^2)^{N_f} } \label{intDnhankel}~. \qquad \qquad
\eea}

\paragraph{Integrals over the intervals $[0,1]^n$.}  Let us change the variable
\bea
y_a = 2 \zeta_a-1~. 
\eea
Therefore, we have
{\small
\bea
g_{N_f, D_n} (t) 
&=& \frac{2^{(n-1)^2+n^2}}{(2 \pi)^n n!}  \frac{1}{(-4t)^{n N_f}} \int \limits_{0}^1 \ud \zeta_1 \cdots \int \limits_{0}^1 \ud \zeta_n \prod_{1 \leq a < b \leq n} (\zeta_a - \zeta_b)^2 \prod_{a=1}^n  \frac{\zeta_a^{-1/2} (1- \zeta_a)^{-1/2}}{ \left(\zeta_a- \frac{(1+t)^2}{4t} \right)^{N_f} } \nn \\
&=& C_{N_f,D_n}(t) D_{N_f,D_n} (T)~, \label{HShankelDn}
\eea}
where the factor $C_{N_f,D_n}$ is given by
\bea
C_{N_f,D_n}(t) :=\frac{2^{(n-1)^2+n^2}}{(2 \pi)^n} \frac{1}{(-4 t)^{n N_f}}~,
\eea
and $D_{N_f,D_n} (T)$ and the variable $T$ are given by
\bea
T &:=& \frac{(1+t)^2}{4t}~, \\
D_{N_f,D_n} (T) &:=& \frac{1}{n!} \int \limits_{0}^1 \ud \zeta_1 \cdots \int \limits_{0}^1 \ud \zeta_n \prod_{1 \leq a < b \leq n} (\zeta_a - \zeta_b)^2 \prod_{k=1}^n w(\zeta_a;T)  \nn\\
&=& \det \left( \int_0^1 \ud \zeta~w(\zeta;T) \zeta^{i+j} \right)_{i,j =0}^{n-1}~, \label{Hankelsoeven} \\
w(\zeta_a;T) &:=&   \zeta_a^{-1/2} (1- \zeta_a)^{-1/2}(\zeta_a - T)^{-N_f}~.
\eea
Indeed, $D_{N_f,D_n} (T)$ is the {\it Hankel determinant} with the perturbed Jacobi weight
\bea
w(\zeta;T) &:=&   \zeta^{\alpha} (1- \zeta)^{\beta}(\zeta - T)^{\gamma}~,
\eea
with the parameters
\bea
\alpha = -1/2, \qquad \beta= -1/2, \qquad \gamma =- N_f~. \label{abcDn}
\eea

\paragraph{Palindromic numerator of the Hilbert series.}  Observing that $T$ is invariant under the transformation $t \mapsto 1/t$, it is clear that
\bea
g_{N_f, D_n} (1/t) =  t^{2n N_f} g_{N_f, D_n} (t)=   (-t)^{2n N_f} g_{N_f, D_n} (t)= (-t)^{N_c N_f} g_{N_f, D_n} (t)~;
\eea
in other words, the numerator of the Hilbert series is palindromic.  Note that the above equalities are in the same form as \eref{invertBn}.

\subsection{The EXDT IV formula} 
In this section, we compute the integrals \eref{gNfDnunref} using the EXDT IV formula.   Indeed, the factor $(1- y_a)^{-1/2}(1+y_a)^{-1/2}$ in \eref{intDnhankel} implies that the Hankel determinant \eref{Hankelsoeven} falls into Case IV of Propositions 3.1 and 3.3 of \cite{BE2008} (see also Page 16 of \cite{PJW2}).  Subsequently, we apply Proposition 4.1 of \cite{BE2008} to compute exact expressions of Hilbert series.

\subsubsection{Computation of the Hankel determinant}
Let us first define various necessary quantities in order to apply the EXDT IV formula.

\paragraph{The symbol and its factorisation.} The symbol for our problem is
\bea
a(z) :=  (1-t z)^{-N_f} (1-t/z)^{-N_f}~. \label{symbDn}
\eea
Put $a(z) = a_+(z) \ta_+(z)$ with
\bea
a_+(z) =  (1-tz)^{-N_f}~, \qquad \ta_+(z) = (1-t/z)^{-N_f}~.
\eea

\paragraph{Fourier coefficients and related matrices.} 
The Fourier coefficients $(a_+^{-1})_k$ (with $k \in \BZ$) of a function $a_+^{-1}$ are given by
\bea
(a_+^{-1})_k := \frac{1}{2 \pi i} \oint_{|z|=1} \frac{\ud z}{z} z^{-k} a_+^{-1} = \frac{1}{2 \pi i} \oint_{|z|=1} \frac{\ud z}{z} z^{-k}  (1-t z)^{N_f} = {N_f \choose k} (-t)^k~.
\eea
The Fourier coefficients $(z \ta_+)_k$ (with $k \in \BZ$) of a function $z \ta_+$ are given by
\bea
(z \ta_+)_k &:=& \frac{1}{2 \pi i} \oint_{|z|=1} \frac{\ud z}{z} z^{-k} (z \ta_+) \nn \\
&=& \frac{1}{2 \pi i} \oint_{|z|=1} \frac{\ud z}{z} z^{-k+1} (1-t/z)^{-N_f}
= \begin{cases} N_f t & \mbox{if } k = 0  \\ 1  & \mbox{if } k =1 \\ 0 & \mbox{otherwise.} \end{cases}
\eea
The Fourier coefficients $(z a_+^{-1} \ta_+)_k$ (with $k \in \BZ$) of a function $z a_+^{-1} \ta_+$ are given by
\bea
(z a_+^{-1} \ta_+)_k := \frac{1}{2 \pi i} \oint_{|z|=1} \frac{\ud z}{z} z^{-k} (z a_+^{-1} \ta_+) = \frac{1}{2 \pi i} \oint_{|z|=1} \frac{\ud z}{z} z^{-(k-1)}  \left( \frac{1- tz}{1-t/z} \right)^{N_f} = c_{k-1},
\eea
where $c_{k-1}$ is given by \eref{ck2F1}.

\paragraph{The matrices $K^D$ and $K^D_n$.}  From Case IV on Page 13 of \cite{BE2008}, we define an infinite matrix $K^D$ to be the such that the $(i, j)$-entry (with $i, j =0,1,2,\ldots$) is given by
\bea
K^D(i,j) &=& (z a_+^{-1} \ta_+)_{i+j +1} - \sum_{l=0}^i (a^{-1}_+)_{i-l} (z  \ta_+)_{l+j+1}~,
\eea
where the superscript $D$ indicates the gauge group $D_n = SO(2n)$.
We can compute $K^D$ explicitly as follows:
\bea
K^D(i,j) = c_{i+j} - \sum_{l=0}^i {N_f \choose i-l} (-t)^{i-l} \delta_{l+j,0} = c_{i+j} - \delta_{j,0} {N_f \choose i} (-t)^{i}~. \label{Kijdef}
\eea
Let us define the matrix $K^D_n$ to be
\bea
K^D_n = Q_n K^D Q_n~,
\eea
where $Q_n$ is given by \eref{Qndef}.
Since $n \geq 1$, it follows that
\bea
K^D_n (i,j)=  \left\{ 
\begin{array}{l l}
0 &\quad \text{for $1 \leq i, j \leq n-1$ and $i+j \geq N_f+1$}\\
c_{i+j} &\quad \text{otherwise}~.
\end{array} \right.  \label{KD}
\eea


\paragraph{The explicit formula.} The integrals \eref{intDnhankel} can be computed from the following formula (see Proposition 4.1 of \cite{BE2008}):
\bea
g_{N_f, D_n}(t) = G(a)^n \widehat{F}_{IV} (a) \det (\BU + K^D_n) ~,
\eea
where the function $G(a)$ is defined by
\bea
G(a) := \exp \left(\frac{1}{2 \pi} \oint_{|z|=1} \frac{\ud z}{z} \log a(z) \right) = 1~,
\eea
and the function $\widehat{F}_{IV} (a)$ is given by (see Proposition 3.3 of \cite{BE2008}):
\bea
\widehat{F}_{IV} (a) &=& \exp \left( \sum_{n=1}^\infty [\log a]_{2n} +\frac{1}{2} \sum_{n=1}^\infty n [\log a]_n^2 \right) \nn \\
&=&  \exp \left(N_f  \sum_{n=1}^\infty \frac{t^{2n}}{2n} +  \frac{1}{2}N_f^2 \sum_{n=1}^\infty n \times \frac{t^{2n}}{n^2} \right) \nn \\
&=& \exp \left( -\frac{1}{2}N_f \log \left( 1-t^2 \right) - \frac{1}{2} N_f^2 \log(1-t^2) \right) \nn \\
&=& (1-t^2)^{-N_f(N_f +1)/2}~.
\eea

\paragraph{The Hilbert series.} The Hilbert series is then given by
\bea
g_{N_f, D_n} (t) = G(a)^n \widehat{F}_{IV} (a) \det (\BU + K^D_n) =  \frac{\det (\BU + K^D_n)}{(1-t^2)^{N_f(N_f +1)/2}}~. \label{unrefBOBEHSDn}
\eea 

\subsubsection{Explicit examples and asymptotic formula}
Below we give certain explicit examples.

\paragraph{The case of $N_f < 2n$.}  In this case $K^D_n (i,j ) =0$ for all $i, j$.  Therefore, the Hilbert series is
\bea
g_{N_f < 2n} (t) = \frac{1}{(1-t^2)^{N_f(N_f+1)/2}}~. \label{gNfl2n}
\eea

\paragraph{The case of $N_f = 2n$.}  In this case 
\bea
K^D_n (i,j ) = \left\{ 
  \begin{array}{l l}
    t^{N_f} & \quad \text{if $i = j =n$}\\
    0 & \quad \text{otherwise}~.
  \end{array} \right. 
\eea
The Hilbert series is thus
\bea
g_{N_f = 2n, D_n} (t) = \frac{1+t^{N_f}}{(1-t^2)^{N_f(N_f+1)/2}}~. \label{fNfeq2n}
\eea

\paragraph{The case of $N_f = 2n+1$.}  The non-trivial block of the matrix $K^D_n$ is given by
\bea
\left(
\begin{array}{cc}
 (1+2 n) t^{2 n} \left(1-t^2\right) \qquad & -t^{1+2 n} \\
 -t^{1+2 n} & 0
\end{array}
\right)~.
\eea
Therefore, the Hilbert series is
\bea
g_{N_f = 2n+1, D_n} (t) &=& \frac{1+(1+2 n) t^{2 n}-(1+2 n) t^{2+2 n}-t^{2+4 n}}{ (1-t^2)^{N_f(N_f+1)/2} }~. \label{2np1Dn}
\eea

\paragraph{The case of $N_f = 2n+2$.}  The non-trivial block of the matrix $K^D_n$ is given by
{\small
\bea
\left(
\begin{array}{ccc}
 (1+n) t^{2 n} \left(1-t^2\right) \left(1+2 n-(3+2 n) t^2\right) \qquad & 2 (1+n) t^{1+2 n} \left(-1+t^2\right) \qquad & t^{2+2 n} \\
 2 (1+n) t^{1+2 n} \left(-1+t^2\right) & t^{2+2 n} & 0 \\
 t^{2+2 n} & 0 & 0
\end{array}
\right)~. \nn \\
\eea}
Therefore, the Hilbert series is
\bea
g_{N_f = 2n+2, D_n} (t) &=& \frac{1}{(1-t^2)^{N_f(N_f+1)/2}} \Big[ 1+(1 + n) (1 + 2 n) t^{2 n}-(1+2 n) (3+2 n) t^{2+2 n} \nn \\
&& +(1+n) (3+2 n) t^{4+2 n}-(1+n) (3+2 n) t^{2+4 n}+(1+2 n) (3+2 n) t^{4+4 n} \nn \\
&& -(1+n) (1+2 n) t^{6+4 n}-t^{6+6 n} \Big]~. \label{2np2Dn}
\eea
%

\subsubsection*{Comments on the results}  
Let us state some comments on the above results.

\ben
\item Given an $n$ and a $\Delta := N_f-n$, the Hilbert series of $D_n=SO(2n)$ SQCD with $n+\Delta$ flavours can be obtained from the Hilbert series of $B_n=SO(2n+1)$ SQCD with $n+\Delta$ flavours case by setting $n$ in the latter to $n-\frac{1}{2}$, \eg~ the Hilbert series \eref{2np2Dn} can be obtained from \eref{2np3Bn} in such a way.  This fact follows from the definition  \eref{KD} of the matrix $K^D_n$.

\item It follows from Case IV of Proposition 3.1 and Proposition 4.1 of \cite{BE2008} that the Hilbert series is a rational function in $t$.  Furthermore, by considering the transformation property of $a(z)$ in \eref{symbDn} under $t \mapsto 1/t$, it is immediate that the numerator of the Hilbert series is palindromic.
\een

\subsubsection*{Asymptotics of unrefined Hilbert series}
Let us examine the asymptotic behaviour of $\det(\BU + K^D_n)$ as $n \rightarrow \infty$.  Since both $2n$ and $2n+1$ are $O(2n)$ in this limit, one should anticipate that such an asymptotic formula should be equal to that for $\det(\BU + K^B_n)$ given by \eref{asympBn}.  In order to see this, we proceed as follows.

Let $\Delta = N_f - 2n$.  For $\Delta \leq 0$, the exact results are given by \eref{gNfl2n} and \eref{fNfeq2n}.  We shall henceforth suppose that $\Delta >0$.
As in the above examples, we can obtain $\det(\BU + K^D_n)$ from $\det(\BU + K^B_n)$ by setting $n$ in the latter to $n - \frac{1}{2}$.
Thus, from \eref{asympBn}, we obtain
\bea
\det (\BU+K^D_n ) &\sim& 1+ \sum_{k=0}^\Delta (-1)^k {2n + \Delta+k \choose \Delta -k}{ 2n -1+2k \choose k} t^{2n+2k} + O(t^{4n+2}) \nn \\
&=& 1 + t^{N_c} {N_f \choose \Delta} {}_3F_2 \left( \frac{1}{2}N_c,-\Delta ,N_f+1; \frac{1}{2}N_c +1,N_c; t^2 \right) \nn \\
&& + O(t^{2N_c+2})~.   \label{asympdetKDn}
\eea
Observe that the second equality of \eref{asympdetKDn} is equal to the last equality of \eref{asympBn}, as expected.  Therefore, as expected, the asymptotic formula (as $n \rightarrow \infty$) for the Hilbert series of $SO(N_c = 2n)$ SQCD with $N_f$ flavours and $\Delta = N_f - N_c = O(1)$ is the same expression as in  \eref{asympHSBn}:
\bea \label{asympHSDn}
g_{N_f, SO(N_c)} (t) &\sim& \frac{1}{(1-t^2)^{N_f(N_f+1)/2} } \Big[ 1 + t^{N_c} {N_f \choose \Delta}  \times \nn \\
&& {}_3F_2 \left( \frac{1}{2}N_c,-\Delta ,N_f+1; \frac{1}{2}N_c +1,N_c; t^2 \right) + O(t^{2N_c+2}) \Big]~.
\eea




\section{$Sp(n)$ SQCD with $N_f$ flavours} \label{sec:Cn}
In this section, we examine the Hilbert series of SQCD with $C_n = Sp(n)$ gauge group and $N_f$ flavours of quarks. 

\subsection{The computations of Hilbert series}
The Haar measure of $Sp(n)$ is given by
\bea
\int \ud \mu_{C_n} &=&  \frac{1}{(2\pi )^n n!} \oint \limits_{|z_1| =1}  \frac{\ud z_1}{z_1} \cdots \oint \limits_{|z_{n}| =1}  \frac{\ud z_{n}}{z_{n}} \left | \Delta_{n} \left(z + \frac{1}{z} \right) \right|^2  \prod_{a=1}^n \left[ 1- \frac{1}{2} \left(z^2_a + \frac{1}{z^2_a} \right)   \right]~. \nn \\
\eea
The refined Hilbert series for $Sp(n)$ SQCD with $N_f$ flavours can be written as
\bea
g_{N_f, C_n} (t ,x) &=& \int \ud \mu_{C_n} (z_1, \ldots, z_n)   \PE \left[ [1,0, \ldots,0]_x \sum_{a=1}^{n} \left( z_a + \frac{1}{z_a} \right) t \right] \nn \\
&=&  \frac{1}{(2\pi )^n n!} \oint \limits_{|z_1| =1}  \frac{\ud z_1}{z_1} \cdots \oint \limits_{|z_{n}| =1}  \frac{\ud z_{n}}{z_{n}} \left | \Delta_{n} \left(z + \frac{1}{z} \right) \right|^2  \prod_{a=1}^n \left[ 1- \frac{1}{2} \left(z^2_a + \frac{1}{z^2_a} \right)   \right] \nn \\
&& \qquad \times  \PE \left[ [1,0, \ldots,0]_x \sum_{a=1}^{n} \left( z_a + \frac{1}{z_a} \right) t \right]~. \label{inthsrefnfcn}
\eea
Setting $x_1, \ldots, x_n =1$, we obtain the unrefined Hilbert series
\bea
g_{N_f, C_n} (t) 
= \frac{1}{(2\pi )^n n!} \oint \limits_{|z_1| =1}  \frac{\ud z_1}{z_1} \cdots \oint \limits_{|z_{n}| =1}  \frac{\ud z_{n}}{z_{n}} \left | \Delta_{n} \left(z + \frac{1}{z} \right) \right|^2  \prod_{a=1}^n \frac{\left[ 1- \frac{1}{2} \left(z^2_a + \frac{1}{z^2_a} \right) \right]}{(1-t z_a)^{2 N_f} (1-t/z_a)^{2N_f} }~.  \label{gnfCnoint} \nn \\
\eea

\subsection{The Hankel determinant}
In this section, we rewrite the integral form \eref{gnfCnoint} of the unrefined Hilbert series in another way and then recast it in terms of the Hankel determinant.

\paragraph{Integrals over the intervals $[-1,1]^n$.} Writing  $z_a = e^{i \theta_a}$ and $y_a = \cos \theta_a$ in \eref{gnfCnoint}, we obtain
\bea
g_{N_f, C_n} (t) 
&=& \frac{2^{n^2+n}}{(2 \pi)^n n!}\int \limits_{-1}^1 \ud y_1 \cdots \int \limits_{-1}^1 \ud y_n  \prod_{1 \leq a < b \leq n} (y_a - y_b)^2 \prod_{a=1}^n  \frac{(1- y_a)^{1/2}(1+y_a)^{1/2}}{ (1- 2 t y_a +t^2)^{2N_f} }~. \label{mainintCn} 
\eea

\paragraph{Integrals over the intervals $[0,1]^n$.}  We can further rewrite the Hilbert series in terms of the Hankel determinants as follows.  Let us change the variable
\bea
y_a = 2 \zeta_a-1~. 
\eea
Therefore, we have
{\small
\bea
g_{N_f, C_n} (t) 
&=& \frac{2^{2n (n+1)}}{(2 \pi)^n n!}\frac{1}{(-4t)^{2n N_f}} \int \limits_{0}^1 \ud \zeta_1 \cdots \int \limits_{0}^1 \ud \zeta_n \prod_{1 \leq a < b \leq n} (\zeta_a - \zeta_b)^2 \prod_{a=1}^n  \frac{\zeta_a^{1/2} (1- \zeta_a)^{1/2}}{ \left(\zeta_a- \frac{(1+t)^2}{4t} \right)^{2N_f} } \nn \\
&=& C_{N_f,C_n} D_{N_f,C_n} (T)~, \label{HShankelCn}
\eea}
where the factor $C_{N_f,C_n}$ is given by
\bea
C_{N_f,C_n} :=\frac{2^{2n (n+1)}}{(2 \pi)^n}\frac{1}{(4t)^{2n N_f}}~,
\eea
and $D_{N_f,C_n} (T)$ and the variable $T$ are given by
\bea
T &:=& \frac{(1+t)^2}{4t}~, \\
D_{N_f,C_n} (T) &:=& \frac{1}{n!} \int \limits_{0}^1 \ud \zeta_1 \cdots \int \limits_{0}^1 \ud \zeta_n \prod_{1 \leq i < j \leq n} (\zeta_i - \zeta_j)^2 \prod_{k=1}^n w(\zeta_a;T) \nn \\
&=& \det \left( \int_0^1 \ud \zeta~w(\zeta;T) \zeta^{i+j} \right)_{i,j =0}^{n-1}~, \label{HankelSp} \\
w(\zeta_a;T) &:=&   \zeta_a^{1/2} (1- \zeta_a)^{1/2}(\zeta_a - T)^{-2N_f}~. \eea
Indeed, $D_{N_f,C_n} (T)$ is the {\it Hankel determinant} with the perturbed Jacobi weight
\bea
w(\zeta;T) &=&   \zeta^{\alpha} (1- \zeta)^{\beta}(\zeta - T)^{\gamma}~,
\eea
with
\bea
\alpha = 1/2, \qquad \beta= 1/2, \qquad \gamma =- 2N_f~. \label{abcCn}
\eea

\paragraph{Palindromic numerator of the Hilbert series.}  Observing that $T$ is invariant under the transformation $t \mapsto 1/t$. Then it follows from \eref{HShankelCn} that
\bea
g_{N_f, C_n} (1/t) =  t^{4n N_f} g_{N_f, C_n} (t)~;
\eea
in other words, the numerator of the Hilbert series is palindromic. 

%

\subsection{The EXDT III formula}
The integrals \eref{gnfCnoint} can be computed exactly using the EXDT III formula.   The factor $(1- y_a)^{1/2}(1+y_a)^{1/2}$ in \eref{mainintCn} implies that the Hankel determinant \eref{HankelSp} correspond to Case III of Propositions 3.1 and 3.3 of \cite{BE2008} (see also Page 16 of \cite{PJW2}).  Subsequently, we apply Proposition 4.1 of \cite{BE2008} to compute exact expressions of the Hilbert series.

\subsubsection{Computation of the Hankel determinant}
We now apply the EXDT III formula to the integrals \eref{gnfCnoint}.  The necessary quantities for such a formula are defined below.

\paragraph{The symbol and its factorisation.} The symbol for our problem is
\bea
a(z) :=  (1-t z)^{-2N_f} (1-t/z)^{-2N_f}~. \label{symbCn}
\eea
Put $a(z) = a_+(z) \ta_+(z)$ with
\bea
a_+(z) =  (1-tz)^{-2N_f}~, \qquad \ta_+(z) = (1-t/z)^{-2N_f}~.
\eea

\paragraph{Fourier coefficients.} The Fourier coefficients $(z^{-1} a_+^{-1} \ta_+)_k$ (with $k \in \BZ$) of a function $z a_+^{-1} \ta_+$ are defined by
\bea
(z^{-1} a_+^{-1} \ta_+)_k &:=& \frac{1}{2 \pi i} \oint_{|z|=1} \frac{\ud z}{z} z^{-k} (z^{-1} a_+^{-1} \ta_+) \nn \\
&=& \frac{1}{2 \pi i} \oint_{|z|=1} \frac{\ud z}{z} z^{-(k+1)}  \left( \frac{1- tz}{1-t/z} \right)^{2N_f} \nn \\
&=& C_{k+1}~,
\eea
where
\bea
C_k = (-t)^k {2N_f \choose k} {}_2F_1 (k - 2N_f, 2N_f; k+1; t^2)~. \label{coeffCk}
\eea
Note that this is the same as $c_k$ given by \eref{ck2F1}, with $N_f$ replaced by $2N_f$. 

\paragraph{The matrices $K^C$ and $K^C_n$.}  From Case III on Page 13 of \cite{BE2008}, we define an infinite matrix $K^C$ to be the such that the $(i, j)$-entry (with $i, j =0,1,2,\ldots$) is given by
\bea
K^C(i,j) = -(z^{-1} a_+^{-1} \ta_+)_{i+j+1} = -C_{i+j+2}~,
\eea
where the superscript $C$ indicates the gauge group $C_n = Sp(n)$.  We also define the matrix $K^C_n$ to be
\bea
K^C_n = Q_n K^C Q_n~,
\eea
where $Q_n$ is given by \eref{Qndef}.
It follows that
\bea \label{zeroKCn}
K^C_n(i, j) =  \left\{ 
\begin{array}{ll}
0 &\quad \text{for $0 \leq i, j \leq n-1$,  $i+j \geq 2N_f-1$}  \\
-C_{i+j+2} &\quad \text{otherwise}~.
\end{array} \right.
\eea

\paragraph{The EXDT III formula.} The integrals \eref{gnfCnoint} can be computed from the following formula (see Proposition 4.1 of \cite{BE2008}):
\bea
\CI_{N_f, C_n}(t) = G(a)^n \widehat{F}_{III} (a) \det (\BU + K^C_n) ~,
\eea
where the function $G(a)$ is defined by
\bea
G(a) := \exp \left(\frac{1}{2 \pi} \oint_{|z|=1} \frac{\ud z}{z} \log a(z) \right) = 1~,
\eea
and the function $\widehat{F}_{III} (a)$ is given by (see Proposition 3.3 of \cite{BE2008}):
\bea
\widehat{F}_{III} (a) &=& \exp \left( -\sum_{n=1}^\infty [\log a]_{2n} +\frac{1}{2} \sum_{n=1}^\infty n [\log a]_n^2 \right) \nn \\
&=&  \exp \left(-2N_f  \sum_{n=1}^\infty \frac{t^{2n}}{2n} +  \frac{1}{2}(2N_f)^2 \sum_{n=1}^\infty n \times \frac{t^{2n}}{n^2} \right) \nn \\
&=& \exp \left( N_f \log \left( 1-t^2 \right) - 2N_f^2 \log(1-t^2) \right) \nn \\
&=& (1-t^2)^{-N_f(2N_f-1)}~.
\eea

\paragraph{The Hilbert series.} The Hilbert series is then given by
\bea
g_{N_f, C_n} (t) &=& G(a)^n \widehat{F}_{III} (a) \det (\BU + K^C_n) \nn \\
&=&  \frac{\det (\BU + K^C_n)}{(1-t^2)^{N_f(2N_f-1)}}~. \label{unrefBOBEHSCn}
\eea 
Note that, as discussed in \cite{Hanany:2008kn, Gray:2008yu}, the numerator of the unrefined Hilbert series $g_{N_f,C_n}(t)$ is a palindromic polynomial.

\subsubsection*{Some examples}
Below we give certain explicit examples.

\paragraph{The case of $N_f \leq n$.}  In this case $K^C_n (i,j ) =0$ for all $i, j$.  Therefore, the Hilbert series is
\bea
g_{N_f \leq n} (t) = \frac{1}{(1-t^2)^{N_f(2N_f-1)}}~. \label{NfleqnCn}
\eea

\paragraph{The case of $N_f = n+1$.}  In this case 
\bea
K^C_n (i,j ) = \left\{ 
  \begin{array}{l l}
    -t^{2N_f} & \quad \text{if $i = j =n$}\\
    0 & \quad \text{otherwise}~.
  \end{array} \right. 
\eea
The Hilbert series is thus
\bea
g_{N_f = n+1, C_n} (t) = \frac{1-t^{2 N_f}}{(1-t^2)^{N_f(2N_f-1)}}~. \label{HSnp1Cn}
\eea

\paragraph{The case of $N_f = n+2$.}  The non-trivial block of the matrix $K^C_n$ is given by
{\small
\bea
\left(
\begin{array}{ccc}
 -(2+n) t^{2+2 n} \left(1-t^2\right) \left(3-5 t^2+2 n \left(1-t^2\right)\right) \qquad & 2 (2+n) t^{3+2 n} \left(1-t^2\right) \qquad & -t^{4+2 n} \\
 2 (2+n) t^{3+2 n} \left(1-t^2\right) & -t^{4+2 n} & 0 \\
 -t^{4+2 n} & 0 & 0
\end{array}
\right)~. \nn \\
\eea}
Therefore, the Hilbert series is
\bea
g_{N_f = n+2, C_n} (t) &=& \frac{1}{ (1-t^2)^{N_f(2N_f-1)} } \Big[ 1-(2+n) (3+2 n) t^{2+2 n}+(3+2 n) (5+2 n) t^{4+2 n} \nn \\
&& -(2+n) (5+2 n) t^{6+2 n}-(2+n) (5+2 n) t^{6+4 n}+(3+2 n) (5+2 n) t^{8+4 n} \nn \\
&& -(2+n) (3+2 n) t^{10+4 n}+t^{12+6n} \Big]~. \label{np2Cn}
\eea

The Hilbert series \eref{NfleqnCn} and \eref{HSnp1Cn} are in agreement with the ones obtained in \cite{Hanany:2008kn}.  They indicate that the moduli spaces for $N_f \leq n$ are freely generated, and the one for $N_f = n+1$ is a complete intersection.  Note that the unrefined Hilbert series \eref{np2Cn} for general $n$ has not been obtained before in \cite{Hanany:2008kn}.  Other results for larger $\Delta = N_f -n$ can also been obtained in a straightforward way.   

It follows from Case III of Proposition 3.1 and Proposition 4.1 of \cite{BE2008} that the Hilbert series is a rational function in $t$.  Furthermore, by considering the transformation property of $a(z)$ in \eref{symbCn} under $t \mapsto 1/t$, it is immediate that the numerator of the Hilbert series is palindromic.

\subsubsection{Asymptotics of unrefined Hilbert series}
Let us first focus on the numerator $\det (\BU + K^C_n)$ of the Hilbert series \eref{unrefBOBEHSCn}.   The asymptotic formulae for $\det( \BU + K^C_n)$  can be obtained in a similar way to that for $\det( \BU + K^B_n)$, since the coefficients $C_k$ in \eref{coeffCk} differ from the coefficients $c_k$ in \eref{ck2F1} only by changing $N_f$ in the latter to $2N_f$ in the former.

Consider the Taylor expansion of $\det (\BU + K^C_n)$ as in \eref{detformula}. The smallest order term in $1+ \tr K^C_n$ is $O(t^{2 (n-1+2 \Delta)})$ as $n \rightarrow \infty$.  For the terms $\left[ (\tr K^C_n)^2 - \tr ({K^C_n}^2) \right]$, it can be checked that the highest contribution is $O(t^{4n+6})$.  The terms with higher traces are smaller.  

In order that there are no overlaps between $1+ \tr K^C_n$ and $\left[ (\tr K^C_n)^2 - \tr ({K^C_n}^2) \right]$, let us {\it assume} that $\Delta$ is $O(1)$ as $n \rightarrow \infty$.  Then, the coefficients of $t^{2n+1+2k}$, for all $ 0 \leq k \leq 2\Delta-2$, can be extracted from $1+\tr K^B_n$ and the subleading terms can be neglected. 

Let $\Delta = N_f-n$.  Recall from \eref{zeroKCn} that $K^C_n(i, j) =0$ for $0 \leq i, j \leq n-1$,  $i+j \geq 2N_f-1$.   Therefore, we find that
\bea
\tr K^C_n &=& \sum_{l=0}^{\Delta} K^C_n (n+l,n+l)  \nn \\
&=& - \sum_{l=0}^{\Delta} C_{2n+2l+2} \nn \\
&=& - \sum_{l=0}^{\Delta} \sum_{j = 0}^{2\Delta-2l-2}  {2N_{f} \choose  2n+2+2l+ j}  {2N_{f}+ j-1 \choose j}  (-1)^{j+2n+2l+2} t^{2j +2n+2l+2} \nn \\
&=&- \sum_{l=0}^{\Delta} \sum_{j = 0}^{2\Delta-2l-2}  {2n+2\Delta \choose  2\Delta-2-2l- j}  {2n+2\Delta-1 + j \choose j}  (-1)^{j} t^{2j +2n+2l+2}~. \nn \\
\eea
Now let us extract the coefficient of $t^{2n+2+2k}$ (with $ 0 \leq k \leq 2\Delta-2$). This can be obtained when $j =  k-l$ (with the constraint $0 \leq j = k-l  \leq 2\Delta-2l-2$).  Therefore, such a coefficient can be written as
\bea
\FC_{2n+2+2k} =- \sum _{l = 0}^{\min(k,2\Delta-2-k)} {2n+2\Delta \choose  2\Delta-2-k-l}  {2n+2\Delta-1 + k-l \choose k-l}  (-1)^{k - l}~.  \qquad
\eea
Using Lemma \ref{estellesum}, we find that the coefficients $\FC_{2n+2+2k}$ can be rewritten as
\bea
\FC_{2n+2+2k} = (-1)^{k+1} {2 n+2 \Delta +k \choose 2\Delta -2-k}{2 n+2 k +1 \choose k}~,
\eea
for $0 \leq k \leq 2\Delta-2$.  Thus, the leading behaviour of the numerator $\det(1 + K^B_n)$ in the limit $n \rightarrow \infty$ is
\bea
\det (\BU+K^C_n ) &\sim& 1+ \sum_{k=0}^{2\Delta-2} \FC_{2n+1+2k} t^{2n+2+2k} + O(t^{4n+6}) \nn \\
&=& 1+ \sum_{k=0}^{2\Delta-2}  (-1)^{k+1} {2 n+2 \Delta +k \choose 2\Delta -2-k}{2 n+2 k +1 \choose k} t^{2n+2+2k} + O(t^{4n+6}) \nn \\
&=& 1+ \sum_{k=0}^{2\Delta-2} (-1)^{k+1} {2 N_f+k \choose 2\Delta -2 -k}{ 2n +2k +1 \choose k} t^{2n+2+2k} + O(t^{4n+6})~. \nn \\
&=& 1 + t^{2n+2} {2 N_f \choose 2 \Delta-2} {}_3F_2 \left( n+1,2(1- \Delta),2 N_f+1;~ n+2, 2(n+1);~ t^2 \right)  \nn \\ && + O(t^{4n+6})~. \label{asympCn}
\eea

\paragraph{Example: $N_f = n+2$.} We have
\bea
\det (\BU+K^C_n ) &\sim&   1-(2+n) (3+2 n) t^{2+2 n}+(3+2 n) (5+2 n) t^{4+2 n} \nn \\
&& -(2+n) (5+2 n) t^{6+2 n} + O(t^{4n+6}) ~.
\eea
This is in agreement with the `first half' of the numerator of \eref{np2Cn}.

\section{The Painlev\'e VI equation} \label{PainleveVI}
In this section we relate the Hilbert series computed in the preceding sections with the Painlev\'e VI equation.  We know the connection because of the results of \cite{DaiZhang}, where it is shown that certain expressions that involve Hankel determinants satisfy the Painlev\'e VI equation.  Our aims of this section are (i) to write down explicitly the corresponding parameters of the Painlev\'e VI equation to $SO(2n+1)$, $SO(2n)$ and $Sp(n)$ SQCD, (ii) to show that the Hilbert series give rise to an infinite family of rational solutions, with palindromic numerators, to the Painlev\'e VI equation.  We summarise the information for (i) in \tref{tab:parameters}.

We first present the Jimbo-Miwa-Okamoto $\sigma$-form of the Painlev\'e VI equation whose connection with the Hankel determinant is most transparent.  Subsequently we move on to the standard form of the Painlev\'e VI equation -- this is the most common form appearing in the literature.  Below we follow closely the presentation of \cite{DaiZhang}.


\begin{table}[htdp]
\hspace{-1cm}
{\small
\begin{tabular}{|c|c|c|c|c|c|c|c|c|c|c|c|}
\hline
Gauge group & \multicolumn{3}{|c|}{Hankel parameters}  & \multicolumn{4}{|c|}{$\sigma$-form of Painlev\'e VI}  & \multicolumn{4}{|c|}{Standard form of Painlev\'e VI}  \\
\cline{2-4} \cline{5-8} \cline{9-12}
~& $\alpha$ & $\beta$ & $\gamma$ & $\nu_1$ & $\nu_2$ & $\nu_3$ & $\nu_4$ & $\mu_1$ & $\mu_2$ & $\mu_3$ & $\mu_4$  \\
\hline
$B_n =SO(2n+1)$ & $-\frac{1}{2}$ & $\frac{1}{2}$ & $-N_f$ & $0$ & $\frac{1}{2}$ & $n$ & $n-N_f$ & $\frac{1}{2}\Delta_B^2$ & $-\frac{1}{8}$ & $\frac{1}{8}$ & $\frac{1}{2}(1-N_f^2)$  \\
\hline
$D_n =SO(2n)$  & $-\frac{1}{2}$ & $-\frac{1}{2}$ & $-N_f$ & $-\frac{1}{2}$ & $0$ & $n-\frac{1}{2}$ & $n - \frac{1}{2} -N_f$ & $\frac{1}{2}\Delta_D^2$ & $-\frac{1}{8}$ & $\frac{1}{8}$ & $\frac{1}{2} (1-N_f^2)$ \\
\hline
$C_n =Sp(n)$  & $\frac{1}{2}$ & $\frac{1}{2}$ & $-2N_f$ & $\frac{1}{2}$ & $0$ & $n+\frac{1}{2}$ & $n + \frac{1}{2} -2N_f$ & $2(\Delta_C-1)^2$& $-\frac{1}{8}$ & $\frac{1}{8}$ & $\frac{1}{2} (1-4N_f^2)$ \\
\hline
\end{tabular}
}
\caption{Summary of the parameters in the perturbed Jacobi weights of the Hankel determinants, the parameters in $\sigma$-form of the Painlev\'e VI equation and the parameters in the standard form of the Painlev\'e VI equation.  Here $\Delta_B = N_f - (2n+1)$, $\Delta_D = N_f -2n$ and $\Delta_C = N_f -n$.}
\label{tab:parameters}
\end{table}%

\subsection{The $\sigma$-form of the Painlev\'e VI equation}
For the sake of completeness, we summarise the relevant results of \cite{DaiZhang}.  Let $D_n(T)$ be the Hankel determinant
\bea 
D_{n} (T) &:=& \frac{1}{n!} \int \limits_{0}^1 \ud \zeta_1 \cdots \int \limits_{0}^1 \ud \zeta_n \prod_{1 \leq a < b \leq n} (\zeta_a - \zeta_b)^2 \prod_{k=1}^n w(\zeta_a;T)  \nn \\
&=& \det \left( \int_0^1 \ud \zeta~w(\zeta;T) \zeta^{i+j} \right)_{i,j =0}^{n-1}~, \label{integralDn}
\eea
with the perturbed Jacobi weight 
\bea
w(\zeta;T) =   \zeta^{\alpha} (1- \zeta)^{\beta}(\zeta - T)^{\gamma}~.
\eea
Let us define
\bea
H_n (T) &:=& T(T-1) \frac{\ud}{\ud T} \log D_n (T)~, \label{H}  \\
\tH_n (T) &:=& H_n (T) +d_1 T + d_2~, \label{Htilde}
\eea
where
\bea
d_1 &=& -n (n+ \alpha+\beta+\gamma) - \frac{1}{4} (\alpha+\beta)^2~, \nn \\
d_2 &=& \frac{1}{4} \left[ 2n(n+\alpha+\beta+\gamma) + \beta(\alpha+\beta) - \gamma (\alpha-\beta) \right]~.
\eea
 
The results in \cite{DaiZhang} show that the function $\tH (T)$ satisfies the Jimbo-Miwa-Okamoto $\sigma$-form of the Painlev\'e VI equation \cite{JM, Okamoto}:
\bea
&& \tH'_n (T(T-1) \tH''_n )^2 + \left[2 \tH'_n (T \tH'_n - \tH_n) - (\tH'_n)^2 - \nu_1 \nu_2 \nu_3 \nu_4 \right]^2 \nn \\
&& = (\tH'_n + \nu_1^2)(\tH'_n + \nu_2^2)(\tH'_n + \nu_3^2)(\tH'_n+ \nu_4^2)~, \label{sigmap6}
\eea
where a prime ($'$) denotes the derivative with respect to $T$ and the parameters $\nu_1, \ldots, \nu_4$ can be written in terms of the Hankel parameters $\alpha, \beta, \gamma$ as
\bea
\nu_1 = \frac{1}{2}(\alpha+\beta), \quad \nu_2 = \frac{1}{2}(\beta-\alpha), \quad \nu_3= n + \frac{1}{2}(\alpha+\beta), \quad \nu_4 = n + \frac{1}{2} (\alpha + \beta + 2\gamma)~. \label{sigmaformpara}
\eea

We should like to mention here that in 1995 Magnus \cite{Magnus} obtained the so-called Magnus/Schlesinger equation which can be reduced in a special case to \eref{sigmap6}.  In \cite{ChenIts}, a Magnus/Schlesinger equation was derived using the Riemann-Hibert method in the context of the Akhiezer polynomials.

\paragraph{Parameters of the equation.}
Recall that the variable $T$ above is related to the variable $t$ in the Hilbert series as
\bea
T = \frac{(1+t)^{2}}{4t}~. \label{defTt}
\eea
Substituting $\alpha, \beta, \gamma$ from \tref{tab:parameters} into \eref{sigmaformpara} , we can write the parameters $\nu_1, \ldots, \nu_4$ in terms of $N_c$ and $N_f$ as tabulated in \tref{tab:parameters}.

%
%

\paragraph{Direct checks.} It can also be directly checked that the Hilbert series computed in the preceding sections give rise to the Hankel determinants and hence the functions $\tH_n(T)$ which satisfy the Painlev\'e VI equation \eref{sigmap6}.  Moreover, we perform a similar check for each of the asymptotic formulae and find that the corresponding function $\tH_{n}(T)$ satisfy the Painlev\'e VI equation \eref{sigmap6} when both sides are expanded as power series of $t$ up to the order of the remainder term in such an asymptotic formula.  

\subsection{Infinite families of rational solutions with palindromic numerators} \label{sec:solutionsP6}
Let us define 
\bea
\CD_n(t) := D_n \left(\frac{(1+t)^2}{4t} \right), \quad \CH_n(t) := H_n \left(\frac{(1+t)^2}{4t} \right), \quad \tCH_n(t) := \tH_n \left(\frac{(1+t)^2}{4t} \right)~.
\eea
These functions are simply the aforementioned $D_n(T)$, $H_n(T)$ and $\tH_n(T)$ with the change of variable $T = \frac{(1+t)^{2}}{4t}$.  Note that $\tCH_n(t)$ is a solution to the Painlev\'e VI equation \eref{sigmap6} when $T$ written in terms of $t$.

In this subsection, we show that the solution $\tCH_n(t)$ is a rational function in $t$ with a palindromic numerator.  Hence, for all possible values of $N_c$ and $N_f$, the Hilbert series of $SO(N_c)$ and $Sp(N_c)$ SQCD with $N_f$ flavours give rise to {\it infinite families of rational solutions, with a palindromic numerators, to the Painlev\'e VI equation}.

Since the Hilbert series is a rational function, it is clear from \eref{HShankelBn}, \eref{HShankelDn} and \eref{HShankelCn} that $\CD_{n}(t)$ is also a rational function. Therefore it is clear that $\CH_n(t)$ and $\tCH_n(t)$ are also rational functions in $t$.

Now we show that the numerator of $\tCH_n(t)$ is palindromic.  We make use of the fact that $T$ is invariant under the transformation $t \mapsto 1/t$.  Therefore the functions $D_n(T)$, $H_n(T)$ and $\tH_n(T)$ are also invariant under such a transformation.  Hence the function $\tCH_n(t)$ has the following property:
\bea
\tCH_n(1/t) = \tCH_n(t)~.
\eea
Hence, the function $\tCH_n(t)$ has a palindromic numerator.

\subsection{The standard form of the Painlev\'e VI equation} 
According to Theorem 1.2 of \cite{DaiZhang},\footnote{Equation (5.1) of \cite{DaiZhang} should read
\bea
r^*_n &=& \frac{1}{2 t R_n} \Big[ \beta  (1+2 n+\alpha +\beta +\gamma )-(1+2 n+\alpha +2 \beta -t (1+\alpha +\beta )+\gamma ) R_n   \nn \\
&& \quad +(1-t) R_n^2+2 r_n (1+2 n+\alpha +\beta +\gamma -(1-t) R_n)-(1-t) t R_n' \Big]~.\nn
\eea
} the $\sigma$-form \eref{sigmap6} is related to the standard form of the Painlev\'e VI equation as follows.  Let
\bea
W_n (T) = \frac{(T-1)R_n(T)}{2n+\alpha+\beta+\gamma+1}+1~,
\eea
where $R_n(T)$ is given by (4.6) of \cite{DaiZhang}:
\bea
R_n(T) = \frac{2(2n+1+\alpha+\beta+\gamma)(\beta+r_n)r_n}{l(r_n,r^*_n,T)+T(1-T)r'_n(T)}~, \label{RnT}
\eea
and 
\bea
r_n &=& \frac{n(n+\alpha+\gamma)-T H_n'+ H_n}{2n + \alpha + \beta+ \gamma}~, \\
r_n^*&=& -\frac{n(n+\beta+\gamma)+ (T-1) H_n' - H_n}{2n + \alpha+ \beta + \gamma}~, \\
l(r_n, r_n^*,T) &:=& 2(1-T)r_n^2+ \big[(2n-\beta+\gamma)T+2\beta+2 T r_n^* \big]r_n \nn \\
&& -(2n+\alpha+\gamma) T r_n^*-n(n+\gamma) T~.
\eea
We use to the prime ($'$) to denote a derivative with respect to $T$.
Then, the function $W_n (T)$ was shown in \cite{Magnus} to satisfy a particular Painlev\'e VI equation:
\bea
W''_n &=& \frac{1}{2}\left( \frac{1}{W_n} + \frac{1}{W_n-1} + \frac{1}{W_n - T}  \right) (W'_n)^2 - \left( \frac{1}{T} + \frac{1}{T-1} + \frac{1}{W_n - T} \right) W'_n \nn \\
&&  + \frac{W_n(W_n-1)(W_n - T)}{T^2( T^2-1)} \left( \mu_1 + \mu_2 \frac{T}{W^2_n} + \mu_3 \frac{T-1}{(W_n-1)^2} + \mu_4 \frac{T(T-1)}{(W_n-T)^2} \right)~, \label{P6std}\qquad
\eea 
with the parameters given by (see also (1.16) of \cite{DaiZhang})
\bea
\mu_1 = \frac{(2n+\alpha+\beta+\gamma+1)^2}{2}, \qquad \mu_2 = - \frac{\alpha^2}{2}, \qquad \mu_3 = \frac{\beta^2}{2}, \qquad \mu_4 = \frac{1-\gamma^2}{2}~.  \label{P6sigmaformpara}
\eea
The parameters $\mu_1, \ldots, \mu_4$, written in terms of $N_c$ and $N_f$, are tabulated in \tref{tab:parameters}.

%
%
%

\subsubsection*{Rational solutions with palindromic numerators}  
As before, using the fact that $T$ is invariant under $t \mapsto 1/t$, we see that $W_n(T)$ is also invariant. Thus, by defining
\bea
\CW_n( t) = W_n\left( \frac{(1+t)^2}{4t} \right)~,
\eea
it follows that
\bea
\CW_n( 1/t) = \CW_n (t)~. \label{CWt}
\eea
From \eref{RnT}, it is clear that $\CW_n (t)$ is a rational function.  Hence \eref{CWt} implies that $\CW_n (t)$ has a palindromic numerator.   Since $W_n(T)$ is a solution of the Painlev\'e VI equation \eref{P6std}, it follows that $\CW_n (t)$ is also a solution to such an equation written in terms to $t$.  The set of all possible values of $n$ and $N_f$ thus leads to infinite families of rational solutions with palindromic numerators.

\section{Integrable systems and elliptic curves} \label{sec:integrability}
In this section, we discuss the Hamiltonian systems associated with the Painlev\'e VI equations previously obtained.  Since those Painlev\'e equations admit solutions arisen from the Hilbert series, such Hamiltonian systems describe the moduli spaces of $SO(N_c)$ and $Sp(N_c)$ SQCD.  

To each Hamiltonian system, we write down the corresponding family of elliptic curves.  As pointed out in \cite{Noumi}, these curves take the form as the Seiberg--Witten curve for $4$d $\CN=2$ $SU(2)$ SQCD with 4 flavours, with the parameters being functions of $N_c$, $N_f$ and fugacity $t$.  

The presence of the Painlev\'e VI equations implies the existence of a Lax pair and hence the integrability of the aforementioned Hamiltonian system.  We end this section by briefly discussing the validity of our results on the quantum moduli space.

\subsection{Integrable Hamiltonian systems}
It is well-known that each of the six Painlev\'e equations is equivalent to a Hamiltonian system (see \eg,~ \cite{bluebook}).  Since we have shown that the Hilbert series of the moduli spaces of $SO(N_c)$ and $Sp(N_c)$ SCQD satisfy Painvlev\'e VI equations, it is possible to write down the explicit Hamiltonian systems that describe such moduli spaces.

The Painlev\'e VI equation \eref{P6std} can be represented by the following Hamiltonian (see \eg,~ \cite{Noumi} and \cite{bluebook})\footnote{We follow the notation of \cite{Noumi} with the following changes of variables: $H \rightarrow H_n$, $f \rightarrow W_n$, $g \rightarrow V_n$ and $s \rightarrow T$. We also set $\delta =1$.}:
\bea
H_n &=& W_n (W_n-1)(W_n - T) V_n^2 + \big[ (a_1+2a_2)(W_n-1)W_n + a_3(T-1) W_n \nn \\
&& +a_4 T(W_n-1) \big] V_n  +a_2 (a_1 +a_2) (W_n-1)~,  \label{Hamiltonian}
\eea
with the Hamiltonian differential equations:
\bea
T( T-1) \frac{\ud W_n}{\ud T} &=& \frac{\partial H_n}{\partial V_n}~, \label{Hamil1} \\
T( T-1) \frac{\ud V_n}{\ud T} &=& - \frac{\partial H_n}{\partial W_n}~. \label{Hamil2}
\eea

\paragraph{Parameters.} In order to determine $a_1, \ldots, a_4$ in terms of the known parameters, we need to obtain the connection between \eref{Hamiltonian} and the Painlev\'e VI equation \eref{P6std}. We proceed as follows.   
\ben
\item Use \eref{Hamil1} to compute $V_n(T)$ in terms of $W_n(T)$ and $W'_n(T)$.  
\item Use \eref{Hamil2} to compute $V'_n(T)$ in terms of $W_n(T)$ and $V_n(T)$.  We substitute $V_n(T)$ from Step 1 into this and obtain another expression of $V'_n(T)$ in terms of $W_n(T)$ and $W'_n(T)$.
\item Take the derivative of $V_n(T)$ obtained in Step 1 and equate this to $V'_n(T)$ obtained in Step 2. We finally arrive at the Painlev\'e VI equation \eref{P6sigmaformpara} with the parameters $\mu_1, \ldots, \mu_4$ given by
\bea
\mu_1 = \frac{1}{2}a_1^2, \qquad \mu_2 = -\frac{1}{2}a_4^2,  \qquad \mu_3 = \frac{1}{2}a_3^2, \qquad  \mu_4 = \frac{1}{2}(1-a_0^2)~,
\eea
where
\bea
a_0 = 1 - a_1 -2a_2 -a_3-a_4~.
\eea
This is in agreement with Theorem 2.1 of \cite{Noumi}.
\een

Using \eref{P6sigmaformpara}, we can write $a_1, \ldots a_4$ in terms of the Hankel parameters $\alpha, \beta, \gamma$ given in \tref{tab:parameters} as follows:
\bea
a_1 = 2n+\alpha+\beta+\gamma+1,  \quad a_2= -(n+\alpha +\beta +\gamma) , \quad a_3 = \beta, \quad a_4 = \alpha~. 
\eea
We explicitly tabulate the parameters $a_1, \ldots, a_4$ in terms of $N_c$ and $N_f$ below.

\begin{table}[htdp]
\begin{center}
\begin{tabular}{|c|c|c|c|c|}
\hline
Gauge group & \multicolumn{4}{|c|}{Parameters of the Hamiltonian} \\
\cline{2-5}
& $a_1$ & $a_2$ & $a_3$ & $a_4$ \\
\hline
 $B_n=SO(2n+1)$  & $-\Delta_B$ & $N_f -n$ & $\frac{1}{2}$ & $-\frac{1}{2}$ \\
 \hline
 $D_n = SO(2n)$ & $-\Delta_D$ & $N_f-n+1$ & $-\frac{1}{2}$ & $-\frac{1}{2}$ \\
 \hline
  $C_n = Sp(n)$ & $2-2 \Delta_C$ & $2N_f-n-1$ & $\frac{1}{2}$ & $\frac{1}{2}$ \\
  \hline
\end{tabular}
\end{center}
\caption{Parameters $a_1, \ldots, a_4$ of the Hamiltonian written in terms of $N_c$ and $N_f$.  Here $\Delta_B = N_f - (2n+1)$, $\Delta_D = N_f -2n$ and $\Delta_C = N_f -n$.}
\label{tab:paraHamil}
\end{table}%

\paragraph{Integrability.} The Lax pair, written in various forms, of the Hamiltonian system \eref{Hamiltonian} associated with the Painlev\'e VI equation is given by, \eg, Eqs. (A.45.8), (A.45.26)--(A.45.30) of \cite{Mehta}, \cite{Conte}, and Eq. (35) of \cite{NoumiLaxPairs}.  Such a Lax pair gives rise to the integrability structure of the Hamiltonian system \eref{Hamiltonian} which describes the moduli spaces of $SO(N_c)$ and $Sp(N_c)$ SQCD.

\subsection{Elliptic curves} \label{sec:curves}
It was pointed out in Appendix A of \cite{Noumi} that each Painlev\'e equation can be associated with Seiberg--Witten curves appearing in 4d $\CN=2$ SQCD with the gauge group $SU(2)$ \cite{Seiberg:1994aj}.  In this section, we write down the corresponding family of elliptic curves to each Painlev\'e VI equation previously obtained.  

According to Table 2 of \cite{Noumi}, the family of elliptic curves corresponding to the Painlev\'e VI equation can be identified with the Seiberg--Witten curves for $\CN=2$ $SU(2)$ gauge theory with 4 flavours.  Subsequently, we follow the notation of \cite{Noumi}, which is equivalent to (16.38) and (17.58) of \cite{Seiberg:1994aj}\footnote{In order to transfer from the notation in (16.38) of \cite{Seiberg:1994aj} to the notation of \cite{Noumi}, one simply shifts $x \rightarrow x+c_1u$ and defines $\rho = -(c_1+c_2),~\sigma=-(c_1-c_2)$.}:
\bea
y^2 &=& x(x- \rho u)(x-\sigma u) - \frac{1}{4} (\rho-\sigma)^2 u_2 x^2 \nn \\
&&  - \left( \frac{1}{4}(\rho-\sigma)^2 \rho \sigma u_4 - \frac{1}{2} \rho \sigma (\rho^2 - \sigma^2) s_4 \right) x \nn \\
&& -(\rho - \sigma) \rho^2 \sigma^2 s_4 u - \frac{1}{4} (\rho - \sigma)^2 \rho^2 \sigma^2 u_6~,
\eea
where 
\bea
\begin{array}{llll}
u_2 &= \sum_{i=1}^4 m_i^2~, \qquad u_4 &= \sum_{1 \leq i <j \leq 4} m_i^2 m_j^2~, \qquad u_6 &= \sum_{1 \leq i < j < k \leq 4} m_i^2 m_j^2 m_k^2~, \nn \\
s_4 &= \prod_{i=1}^4 m_i~, \qquad \rho &= -\theta_3(0,\tau)^4~,  \qquad \qquad \sigma &= - \theta_2(0,\tau)^4~,
\end{array}
\eea
with the theta functions defined as
\bea
\theta_2(0, \tau) = \sum_{n \in \BZ} (-1)^n q^{\frac{1}{2}n^2}~, \qquad \theta_3(0, \tau) = \sum_{n \in \BZ} q^{\frac{1}{2}n^2}~,\qquad  q = \exp(2 \pi i \tau)~.
\eea

The parameters $m_1, \ldots, m_4$ are related to the parameters $a_1, \ldots, a_4$ of \eref{Hamiltonian} and the Hankel parameters $\alpha, \beta, \gamma$ given in \tref{tab:parameters} as follows \cite{Noumi}:
\bea
\begin{array}{lll}
 m_1 &= \frac{1}{2} (a_1+2 a_2+2 a_3+a_4) &=  \frac{1}{2} (1+\beta -\gamma ) \\
m_2 &= \frac{1}{2} (a_1 + 2 a_2 + a_4) &= \frac{1}{2} (1-\beta -\gamma )  \\
m_3 &= \frac{1}{2} (a_1+a_4) &= \frac{1}{2} (1+2 n+2 \alpha +\beta +\gamma ) \\
m_4 &= \frac{1}{2} (a_1 - a_4) &= \frac{1}{2} (1+2 n+\beta +\gamma )~.
\end{array}
\eea
Moreover, the parameter $\tau$ is related to $T = \frac{(1+t)^2}{4t}$ as follows \cite{Noumi}:
\bea
T = \frac{\sigma}{\rho} = \left( \frac{\theta_2(0,\tau)}{\theta_3(0,\tau)} \right)^4~.
\eea

In $\CN=2$ gauge theory, one interprets $m_1, \ldots, m_4$ as the mass parameters and $\tau$ as the gauge coupling parameter.  However it is not clear from our discussion whether such interpretations still hold in $\CN=1$ SQCD we are considering.  At the moment, what we can infer is that the parameters in the curves for $\CN=1$ SQCD are certain functions of $N_c$, $N_f$ and the fugacity $t$.   We leave the issues of the physical origin and physical interpretation of such curves for future works.\footnote{We mention, {\it en passant}, that there are a number of works on $\CN=1$ Seiberg--Witten curves. Many of these are listed in \cite{Tachikawa:2011ea}, namely~ \cite{Intriligator:1994sm,Kapustin:1996nb,Kitao:1996mb,Giveon:1997gr,Csaki:1997zg,Gremm:1997sz,Lykken:1997gy,Giveon:1997sn,deBoer:1997zy,Burgess:1998jh,Csaki:1998dp,Hailu:2002bg,Hailu:2002bh}. It would be interesting to find out a connection between the curves in this paper and the curves in those references.}  

We tabulate the parameters $m_1, \ldots, m_4$ in terms of $N_c$ and $N_f$ in \tref{tab:paracurves}.

\begin{table}[htdp]
\begin{center}
\begin{tabular}{|c|c|c|c|c|}
\hline
Gauge group & \multicolumn{4}{|c|}{Parameters of the elliptic curves} \\
\cline{2-5}
& $m_1$ & $m_2$ & $m_3$ & $m_4$ \\
\hline
 $B_n=SO(2n+1)$  & $\frac{1}{2} \left( \frac{3}{2} +N_f \right)$ & $\frac{1}{2} \left( \frac{1}{2} +N_f \right)$ & $-\frac{1}{2} \left(\Delta_B +\frac{1}{2} \right)$ & $-\frac{1}{2} \left(\Delta_B -\frac{1}{2} \right)$ \\
 \hline
 $D_n = SO(2n)$ & $\frac{1}{2} \left( \frac{1}{2} +N_f \right)$ & $\frac{1}{2} \left( \frac{3}{2} +N_f \right)$ & $-\frac{1}{2} \left(\Delta_D +\frac{1}{2} \right)$ & $-\frac{1}{2} \left(\Delta_D -\frac{1}{2} \right)$ \\
 \hline
  $C_n = Sp(n)$ & $\frac{3}{4}+N_f$ & $\frac{1}{4}+N_f$ & $\frac{5}{4}-\Delta_C$ & $\frac{3}{4}-\Delta_C$ \\
  \hline
\end{tabular}
\end{center}
\caption{Parameters $m_1, \ldots, m_4$ of the elliptic curves written in terms of $N_c$ and $N_f$.  Here $\Delta_B = N_f - (2n+1)$, $\Delta_D = N_f -2n$ and $\Delta_C = N_f -n$.}
\label{tab:paracurves}
\end{table}%

%
%

\subsection{Remarks on the validity of results on quantum moduli spaces} \label{sec:validity}
So far we have focused on the classical moduli spaces of $SO(N_c)$ and $Sp(N_c)$ SQCD with $N_f$ flavours. Indeed, the Hilbert series in the preceding sections and in \cite{Hanany:2008kn} have been computed to characterise such classical moduli spaces.   We thus previously concluded that the classical moduli spaces are described by the integrable Hamiltonian systems given by \eref{Hamiltonian}.

However, as discussed in \cite{Intriligator:1995id, hep-th/9505006}, the moduli spaces in general receive quantum corrections.   It was shown that, for $SO(N_c)$ SQCD with $N_f < N_c - 4$ and $Sp(N_c)$ SQCD with $N_f \leq N_c$, the moduli spaces are totally lifted by dynamical generated superpotential and hence there are no supersymmetric vacua.  Nevertheless, for $SO(N_c)$ SQCD with $N_f \geq N_c -4$ and $Sp(N_c)$ SQCD with $N_f > N_c$, there are certain branches of the quantum moduli spaces on which there are no superpotentials generated and there remain degenerate quantum vacua.  

Nevertheless, in the case that a quantum moduli space still exists, the the Hilbert series computed in this paper still provide valid descriptions of a region far away from singularities.  This is because the generators and the relations are unaffected by quantum corrections and there are no extra massless degrees of freedom appearing \cite{Intriligator:1995id, hep-th/9505006}.  We therefore conjecture that such a region of the quantum moduli space is still described by the aforementioned integrable Hamiltonian system.

 \acknowledgments
N.~M. thanks Ben Hoare for a number of useful discussions. He is also grateful to Amihay Hanany for educating him about Hilbert series and pointing out its significance in supersymmetric gauge theories.  This work is supported by a research grant of the Max Planck Society.

\appendix
\section{Refined Hilbert series} \label{app:refined}
In this Appendix, we compute the refined Hilbert series of $SO(N_c)$ and $Sp(N_c)$ SQCD with $N_f$ flavours.

\subsection{$SO(2n+1)$ SQCD with $N_f$ flavours} \label{AppBnref}
We are interested in computing the following integrals: 
\bea
\CI_{N_f, B_n}(t,x) &=&  \frac{1}{(2 \pi i)^{n} n!} \oint \limits_{|z_1| =1}  \frac{\ud z_1}{z_1} \cdots \oint \limits_{|z_{n}| =1}  \frac{\ud z_{n}}{z_{n}} \left | \Delta_{n} \left(z + \frac{1}{z} \right) \right|^2 \prod_{a=1}^n \left[ 1- \frac{1}{2} \left(z_a + \frac{1}{z_a} \right)   \right] \nn \\
&& \times \PE \left[ [1,0, \ldots,0]_x \sum_{a=1}^n \left( z_a + \frac{1}{z_a} \right) t \right] ~. \label{intNfBnref}
\eea

\paragraph{The symbol and its factorisation.}  We define the symbol $a(z)$ to be
\bea
a(z) := \PE[ [1,0, \ldots,0]_x z t] \PE[ [1,0, \ldots,0]_x z^{-1} t]~. 
\eea
Put $a(z) = a_+(z) \ta_+(z)$ with
\bea
a_+(z) &=&   \PE[ [1,0, \ldots,0]_x z t], \qquad \ta_+(z) =  \PE[ [1,0, \ldots,0]_x z^{-1} t] ~.
\eea
Define the function $c(z)$ as
\bea
c(z) = a_+^{-1} (z) \ta_+ (z) = \frac{ \PE[ [1,0, \ldots,0]_x z^{-1} t]}{ \PE[ [1,0, \ldots,0]_x z t]}~.
\eea

\paragraph{Fourier coefficients.} The Fourier coefficients $c_k$ (with $k \in \BZ$) of a function $c(z)$ are defined by
\bea
c_k := \frac{1}{2 \pi i} \oint_{|z|=1} \frac{\ud z}{z} z^{-k} c(z) 
= \frac{1}{2 \pi i} \oint_{|z|=1} \frac{\ud z}{z} z^{-k}  \frac{ \PE[ [1,0, \ldots,0]_x z^{-1} t]}{ \PE[ [1,0, \ldots,0]_x z t]}~. \label{ckrefBn}
\eea
Note that the coefficients $c_k$ can be computed from the first equation in (2.50) of \cite{Chen:2011wn}:
\bea
c_k &=&  \sum_{m=0}^{N_f-k} [m, 0, \ldots, 0]_x [0, \ldots,0,1_{(k+m);L},0, \ldots,0]_x (-1)^{m+k} t^{2m+k} \nn \\
&=& \sum_{m=0}^{N_f-k} \big( [m, 0,\ldots,0,1_{(k+m);L},0, \ldots,0] + [m-1, 0,\ldots,0,1_{(k+m+1);L},0, \ldots,0] \big) \times \nn \\
&& \qquad \qquad (-1)^{m+k} t^{2m+k}~.  \label{ckref}
\eea
Define an infinite matrix $K^B$ to be the such that the $(i, j)$-entry (with $i, j =0,1,2,\ldots$) is given by
{\small
\bea
 K^B(i,j) &=& - c_{i+j +1} \nn \\
&=& \sum_{m=0}^{N_f-(i+j+1)} \Big( [m, 0,\ldots,0,1_{(i+j+m+1);L},0, \ldots,0] \nn \\
&& \qquad \qquad \quad + [m-1, 0,\ldots,0,1_{(i+j+m+2);L},0, \ldots,0] \Big) (-1)^{m+i+j} t^{2m+i+j+1}~.
\eea}
Take the matrix $K^B_n$ to be
\bea
K^B_n = Q_n K^B Q_n~,
\eea
where $Q_n$ is defined as in \eref{Qndef}.
It follows that
\bea
K^B_n(i, j) =\left\{ 
  \begin{array}{l l}
0 &\quad \text{for}~ 0 \leq i, j \leq n-1~\text{and}~ i+j \geq N_f\\
-c_{i+j+1} &\quad \text{otherwise}~.
\end{array} \right. \label{KBnref}
\eea

\paragraph{EXDT II formula.} The integrals \eref{intNfBnref} can be computed from the EXDT II formula (see Proposition 4.1 of \cite{BE2008}):
\bea
\CI_{N_f, B_n}(t,x) = G(a)^n \widehat{F}_{II} (a) \det (\BU + K^B_n) ~,
\eea
where the function $G(a)$ is defined by
\bea
G(a) := \exp \left(\frac{1}{2 \pi} \oint_{|z|=1} \frac{\ud z}{z} \log a(z) \right) = 1~,
\eea
and the function $\widehat{F}_{II} (a)$ is given by (see Proposition 3.3 of \cite{BE2008}):
\bea
\widehat{F}_{II} (a) &=& \exp \left(- \sum_{n=0}^\infty [\log a]_{2n+1} +\frac{1}{2} \sum_{n=1}^\infty n [\log a]_n^2 \right)~.
\eea
Note that 
\bea
\log a = \sum_{m=1}^\infty \frac{1}{m} \left( [1,0, \ldots,0]_{x^m} z^m t^m +[1,0, \ldots,0]_{x^m} z^{-m} t^m \right) ~.
\eea
Therefore, we obtain
\bea
[\log a]_{2n+1} =  \frac{1}{2n+1} [1,0, \ldots,0]_{x^{2n+1}}  t^{2n+1}~,  \quad
\left[ \log a \right]_n =  \frac{1}{n} [1,0, \ldots,0]_{x^n}  t^n~.
\eea
Thus, we have
\bea
\widehat{F}_{II} (a) &=& \exp \Bigg(- \sum_{n=0}^\infty \frac{1}{2n+1} [1,0, \ldots,0]_{x^{2n+1}}  t^{2n+1}  + \sum_{n=1}^\infty \frac{1}{2n} [1,0, \ldots,0]^2_{x^n}  t^{2n} \Bigg) \nn \\
&=& \exp \Bigg(- \sum_{n=1}^\infty \frac{1}{n} [1,0, \ldots,0]_{x^{n}}  t^{n}  + \sum_{n=1}^\infty \frac{1}{2n} \big( [1,0, \ldots,0]^2_{x^n} +[1,0, \ldots,0]_{x^{2n}} \big) t^{2n} \Bigg) \nn \\
&=& \exp \Bigg(- \sum_{n=1}^\infty \frac{1}{n} [1,0, \ldots,0]_{x^{n}}  t^{n}  + \sum_{n=1}^\infty \frac{1}{n}  [2,0, \ldots,0]_{x^n}  t^{2n} \Bigg) \nn \\
&=& \frac{\PE \big[ [2,0, \ldots,0]_x t^2 \big]}{\PE\big[ [1,0, \ldots,0]_x t \big]}~.
\eea

\paragraph{The Hilbert series.} From \eref{HSBnPE} and \eref{IBnref}, the Hilbert series is then given by
\bea
g_{N_f, B_n} (t,x)  &=&  \PE \left[ [1,0, \ldots,0]_x t \right] \CI_{N_f, B_n} (t,x)  \nn \\
&=&  \PE \left[ [1,0, \ldots,0]_x t \right]  G(a)^n \widehat{F}_{II} (a) \det (\BU + K^B_n) \nn \\
&=&  \det (\BU + K^B_n) \PE \big[ [2,0, \ldots,0]_x t^2 \big]~. \label{refBOBEHS}
\eea 

\subsubsection*{Some examples}
Below we give certain explicit examples.

\paragraph{The case of $N_f < 2n+1$.}  In this case $K^B_n (i,j ) =0$ for all $i, j$.  Therefore, the Hilbert series is
\bea
g_{N_f < 2n+1} (t,x) &=& \PE \big[ [2,0, \ldots,0]_x t^2 \big] \nn \\
&=&  \frac{1}{(1-t^2)^{N_f}} \sum_{m_1, \ldots, m_{N_f-1} \geq 0} [2m_1, 2m_2, \ldots, 2m_{N_f-1}]_x~ t^{2\sum_{j=1}^{N_f-1} j m_j }~. \nn \\ \label{BnNfl2np1}
\eea

\paragraph{The case of $N_f = 2n+1$.}  In this case 
\bea
K^B_n (i,j ) = \left\{ 
  \begin{array}{l l}
    t^{N_f} & \quad \text{if $i = j =n$}\\
    0 & \quad \text{otherwise}~.
  \end{array} \right. 
\eea
The Hilbert series is thus
\bea
g_{N_f = 2n+1, B_n} (t,x) &=& \left(1+t^{N_f} \right)\PE \big[ [2,0, \ldots,0]_x t^2 \big] \nn \\
&=& \sum_{m_1, \ldots, m_{2n+1} \geq 0} [2m_1, 2m_2, \ldots, 2m_{2n}]_x~ t^{2\sum_{j=1}^{2n} j m_j+ (2n+1) m_{2n+1}  }~.
\eea

\paragraph{The case of $N_f = 2n+2$.}  The non-trivial block of the matrix $K^B_n$ is given by
\bea
\begin{pmatrix}
 K^B_n (n,n) \quad& -t^{2+2 n} \\
 -t^{2+2 n} & 0
\end{pmatrix}~.
\eea
where
\bea
K^B_n(n,n) 
&=& \sum_{m=0}^{1} \big( [m, 0,\ldots,0,1_{(2n+m+1);L},0, \ldots,0] \nn \\
&& \quad \quad  + [m-1, 0,\ldots,0,1_{(2n+m+2);L},0, \ldots,0] \big) \times (-1)^{m} t^{2m+2n+1} \nn \\
&=&  [0, \ldots, 0, 1] t^{1+2 n}-[1, 0, \ldots ,0] t^{3+2 n}~. 
\eea
Therefore, the Hilbert series is
\bea
g_{N_f = 2n+2, B_n} (t,x) &=& \left( 1+[0, \ldots, 0, 1]_x t^{1+2 n}-[1, 0, \ldots ,0]_x t^{3+2 n}-t^{4+4 n} \right) \PE \big[ [2,0, \ldots,0]_x t^2 \big] \nn \\
&=& \sum_{m_1, \ldots, m_{2n+1} \geq 0} [2m_1, 2m_2, \ldots, 2m_{2n}, m_{2n+1}]_x~ t^{2\sum_{j=1}^{2n} j m_j+ (2n+1) m_{2n+1}  }~. \nn \\
\eea

\paragraph{The case of $N_f = 2n+3$.}  The non-trivial block of the matrix $K_n$ is given by
\bea
\left(
\begin{array}{ccc}
 K^B_n(n,n) & K^B_n(n, n+1) & t^{3+2 n} \\
 K^B_n(n, n+1) & t^{3+2 n} & 0 \\
 t^{3+2 n} & 0 & 0
\end{array}
\right)~,
\eea
where
\bea
K^B_n(n,n) 
&=& \sum_{m=0}^{2} \big( [m, 0,\ldots,0,1_{(2n+m+1);L},0, \ldots,0] \nn \\
&& \quad \quad  + [m-1, 0,\ldots,0,1_{(2n+m+2);L},0, \ldots,0] \big) \times (-1)^{m} t^{2m+2n+1} \nn \\
&=& [0, \ldots,0,1,0] t^{1+2n} - \left( [1,0, \ldots,0,1] + [0,  \ldots,0] \right) t^{3+2n} + [2,0, \ldots,0] t^{5+2n}~. \nn \\
K^B_n(n,n+1)&=& \sum_{m=0}^{1} \big( [m, 0,\ldots,0,1_{(2n+m+2);L},0, \ldots,0] \nn\\ 
&& \quad \quad + [m-1, 0,\ldots,0,1_{(2n+m+3);L},0, \ldots,0] \big) \times (-1)^{m+1} t^{2m+2n+2} \nn \\
&=& -[0, \ldots,0,1] t^{2+2n} +  [1,0, \ldots,0] t^{4+2n}~.
\eea
Therefore, the Hilbert series is
\bea
g_{N_f = 2n+3, B_n} (t,x) &=& \PE \big[ [2,0, \ldots,0]_x t^2 \big] \times \nn \\
&& \Big( 1+  [0,0,1,0]_x t^{1+2n} - [1,0,0,1]_x t^{3+2n} + [2,0,0,0]_x t^{5+2n} \nn \\
&& - [0,0,0,2]_x t^{4+4n} + [1,0,0,1]_x t^{6+4n} - [0,1,0,0]_x t^{8+4n}  - t^{9+6n} \Big) \nn \\
&=&  \sum_{m_1, \ldots, m_{2n+1} \geq 0} [2m_1, 2m_2, \ldots, 2m_{2n}, m_{2n+1},0]_x~ t^{2\sum_{j=1}^{2n} j m_j+ (2n+1) m_{2n+1}  }~. \nn \\ \label{BnNf2np3}
\eea

\paragraph{General formula.}  Note that the results in the above examples are in agreement with (2.29) and (2.30) of \cite{Hanany:2008kn}:
\bea
g_{N_f , B_n} (t,x) &=& \sum_{m_1, \ldots, m_{2n+1} \geq 0} [2m_1, 2m_2, \ldots, 2m_{2n}, m_{2n+1},0, \ldots,0]_x~ t^{2\sum_{j=1}^{2n} j m_j+ (2n+1) m_{2n+1}  }~. \label{genchaexpBn} \nn \\
\eea

%

\subsection{$SO(2n)$ SQCD with $N_f$ flavours}
We are interested in computing the integral \eref{inthsrefnfdn}:
\bea
g_{N_f, D_n} (t ,x)  &=& \frac{2^{-(n-1)}}{(2 \pi i)^{n} n!} \oint \limits_{|z_1| =1}  \frac{\ud z_1}{z_1} \cdots \oint \limits_{|z_{n}| =1}  \frac{\ud z_{n}}{z_{n}} \left | \Delta_{n} \left(z + \frac{1}{z} \right) \right|^2 \nn \\
&& \times \PE \left[ [1,0, \ldots,0]_x \sum_{a=1}^n \left( z_a + \frac{1}{z_a} \right) t \right] ~.
\eea
\paragraph{The symbol and its factorisation.}  We define the symbol $a(z)$ to be
\bea
a(z) := \PE[ [1,0, \ldots,0]_x z t] \PE[ [1,0, \ldots,0]_x z^{-1} t]~. 
\eea
Put $a(z) = a_+(z) \ta_+(z)$ with
\bea
a_+(z) &=&   \PE[ [1,0, \ldots,0]_x z t], \qquad \ta_+(z) =  \PE[ [1,0, \ldots,0]_x z^{-1} t] ~.
\eea

\paragraph{Fourier coefficients.} The Fourier coefficients $(a_+^{-1})_k$ (with $k \in \BZ$) of a function $a_+^{-1}$ are given by
\bea
(a_+^{-1})_k &:=& \frac{1}{2 \pi i} \oint_{|z|=1} \frac{\ud z}{z} z^{-k} a_+^{-1} \nn\\
&=& \frac{1}{2 \pi i} \oint_{|z|=1} \frac{\ud z}{z} z^{-k} \frac{1}{\PE[ [1,0, \ldots,0]_x z t]} = [0,\ldots,0,1_{k;L},0, \ldots, 0] (-t)^k~.
\eea
The Fourier coefficients $(z \ta_+)_k$ (with $k \in \BZ$) of a function $z \ta_+$ are given by
\bea
(z \ta_+)_k &:=& \frac{1}{2 \pi i} \oint_{|z|=1} \frac{\ud z}{z} z^{-k} (z \ta_+) = \begin{cases}  [1,0,\ldots,0]_x t & \mbox{if } k = 0  \\ 1  & \mbox{if } k =1 \\ 0 & \mbox{otherwise.} \end{cases}
\eea
The Fourier coefficients $(z a_+^{-1} \ta_+)_k$ (with $k \in \BZ$) of a function $z a_+^{-1} \ta_+$ are given by
\bea
(z a_+^{-1} \ta_+)_k &:=& \frac{1}{2 \pi i} \oint_{|z|=1} \frac{\ud z}{z} z^{-k} (z a_+^{-1} \ta_+) = c_{k-1}~,
\eea
where the last equality follows from the definition of $c_k$ given in \eref{ckrefBn}.

\paragraph{The matrices $K^D$ and $K^D_n$.} From Case IV on Page 13 of \cite{BE2008}, we define an infinite matrix $K^D$ to be the such that the $(i, j)$-entry (with $i, j =0,1,2,\ldots$) is given by
\bea
K^D(i,j) &=& (z a_+^{-1} \ta_+)_{i+j +1} - \sum_{l=0}^i (a^{-1}_+)_{i-l} (z  \ta_+)_{l+j+1}~.
\eea
By a similar reasoning to the derivation of \eref{KD}, it follows that
\bea
K^D_n (i,j)=  \left\{ 
\begin{array}{l l}
0 &\quad \text{for $1 \leq i, j \leq n-1$ and $i+j \geq N_f+1$}\\
c_{i+j} &\quad \text{otherwise}~.
\end{array} \right. \label{KDnref}
\eea
Recall that an explicit expression of $c_k$ is given in \eref{ckref}.

\paragraph{EXDT IV formula.} The integrals \eref{inthsrefnfdn} can be computed from the following formula (see Proposition 4.1 of \cite{BE2008}):
\bea
g_{N_f, D_n}(t,x) = G(a)^n \widehat{F}_{IV} (a) \det (\BU + K^D_n) ~, \label{HSDnrefBOBE}
\eea
where the function $G(a)$ is defined by
\bea
G(a) := \exp \left(\frac{1}{2 \pi} \oint_{|z|=1} \frac{\ud z}{z} \log a(z) \right) = 1~,
\eea
and the function $\widehat{F}_{IV} (a)$ is given by (see Proposition 3.3 of \cite{BE2008}):
\bea
\widehat{F}_{IV} (a) &=& \exp \left( \sum_{n=1}^\infty [\log a]_{2n} +\frac{1}{2} \sum_{n=1}^\infty n [\log a]_n^2 \right) \nn \\
&=& \exp \Bigg( \sum_{n=1}^\infty \frac{1}{2n} [1,0, \ldots,0]_{x^{2n}}  t^{2n}  + \sum_{n=1}^\infty \frac{1}{2n} [1,0, \ldots,0]^2_{x^n}  t^{2n} \Bigg) \nn \\
&=& \exp \Bigg( \sum_{n=1}^\infty  \frac{1}{n} [2,0, \ldots,0]_x^n t^{2n} \Bigg) \nn \\
&=& \PE \left[ [2,0,\ldots,0]_x t^2 \right]~.
\eea

\paragraph{The Hilbert series.} From \eref{HSDnrefBOBE}, the Hilbert series is then given by
\bea
g_{N_f, D_n} (t,x) &=&  G(a)^n \widehat{F}_{IV} (a) \det (\BU + K^D_n) \nn \\
&=&  \det (\BU + K^D_n) \PE \big[ [2,0, \ldots,0]_x t^2 \big]~.
\eea 

Similar computations can be performed as for \eref{BnNfl2np1}--\eref{BnNf2np3} in the case of $B_n$.  The results are in agreement with (2.29) and (2.30) of \cite{Hanany:2008kn}:
\bea
g_{N_f , D_n} (t,x) &=& \sum_{m_1, \ldots, m_{2n} \geq 0} [2m_1, 2m_2, \ldots, 2m_{2n-1}, m_{2n},0, \ldots,0]_x~ t^{2\sum_{j=1}^{2n-1} j m_j+ (2n) m_{2n}  }~.  \nn \\
\eea

\subsection{$Sp(n)$ SQCD with $N_f$ flavours}
We are interested in computing the integral \eref{inthsrefnfcn}:
\bea
g_{N_f, C_n} (t ,x)
&=&  \frac{1}{(2\pi )^n n!} \oint \limits_{|z_1| =1}  \frac{\ud z_1}{z_1} \cdots \oint \limits_{|z_{n}| =1}  \frac{\ud z_{n}}{z_{n}} \left | \Delta_{n} \left(z + \frac{1}{z} \right) \right|^2  \prod_{a=1}^n \left[ 1- \frac{1}{2} \left(z^2_a + \frac{1}{z^2_a} \right)   \right] \nn \\
&& \qquad \times  \PE \left[ [1,0, \ldots,0]_x \sum_{a=1}^{n} \left( z_a + \frac{1}{z_a} \right) t \right]~.
\eea
\paragraph{The symbol and its factorisation.}  We define the symbol $a(z)$ to be
\bea
a(z) := \PE[ [1,0, \ldots,0]_x z t] \PE[ [1,0, \ldots,0]_x z^{-1} t]~. 
\eea
Put $a(z) = a_+(z) \ta_+(z)$ with
\bea
a_+(z) &=&   \PE[ [1,0, \ldots,0]_x z t], \qquad \ta_+(z) =  \PE[ [1,0, \ldots,0]_x z^{-1} t] ~.
\eea

\paragraph{Fourier coefficients.} The Fourier coefficients $(z^{-1} a_+^{-1} \ta_+)_k$ (with $k \in \BZ$) of a function $z a_+^{-1} \ta_+$ are defined by
\bea
(z^{-1} a_+^{-1} \ta_+)_k &:=& \frac{1}{2 \pi i} \oint_{|z|=1} \frac{\ud z}{z} z^{-k} (z^{-1} a_+^{-1} \ta_+) = C_{k+1}~,
\eea
where
\bea
C_k  &=&  \sum_{m=0}^{2N_f-k} [m, 0, \ldots, 0]_x [0, \ldots,0,1_{(k+m);L},0, \ldots,0]_x (-1)^{m+k} t^{2m+k} \nn \\
&=& \sum_{m=0}^{2N_f-k} \big( [m, 0,\ldots,0,1_{(k+m);L},0, \ldots,0] + [m-1, 0,\ldots,0,1_{(k+m+1);L},0, \ldots,0] \big) \times \nn \\
&& \qquad \qquad (-1)^{m+k} t^{2m+k}~. 
\eea
Note that this is the same as $c_k$ given by \eref{ckref}, with $N_f$ replaced by $2N_f$. 

\paragraph{The matrices $K^C$ and $K^C_n$.}  From Case III on Page 13 of \cite{BE2008}, we define an infinite matrix $K^C$ to be the such that the $(i, j)$-entry (with $i, j =0,1,2,\ldots$) is given by
\bea
K^C(i,j) = -(z^{-1} a_+^{-1} \ta_+)_{i+j+1} = -C_{i+j+2}~.
\eea
Take the matrix $K^C_n$ to be
\bea
K^C_n = Q_n K^C Q_n~,
\eea
where $Q_n$ is defined as in \eref{Qndef}.
It follows that
\bea \label{zeroKCnref}
K^C_n(i, j) =  \left\{ 
\begin{array}{ll}
0 &\quad \text{for $0 \leq i, j \leq n-1$,  $i+j \geq 2N_f-1$}  \\
-C_{i+j+2} &\quad \text{otherwise}~.
\end{array} \right.
\eea

\paragraph{EXDT III formula.} The integrals \eref{inthsrefnfcn} can be computed from the following formula (see Proposition 4.1 of \cite{BE2008}):
\bea
\CI_{N_f, C_n}(t) = G(a)^n \widehat{F}_{III} (a) \det (\BU + K^C_n) ~,
\eea
where the function $G(a)$ is defined by
\bea
G(a) := \exp \left(\frac{1}{2 \pi} \oint_{|z|=1} \frac{\ud z}{z} \log a(z) \right) = 1~,
\eea
and the function $\widehat{F}_{III} (a)$ is given by (see Proposition 3.3 of \cite{BE2008}):
\bea
\widehat{F}_{III} (a) &=& \exp \left( -\sum_{n=1}^\infty [\log a]_{2n} +\frac{1}{2} \sum_{n=1}^\infty n [\log a]_n^2 \right)~.
\eea

Note that 
\bea
\log a = \sum_{m=1}^\infty \frac{1}{m} \left( [1,0, \ldots,0]_{x^m} z^m t^m +[1,0, \ldots,0]_{x^m} z^{-m} t^m \right) ~.
\eea
Therefore, we obtain
\bea
[\log a]_{2n} =  \frac{1}{2n} [1,0, \ldots,0]_{x^{2n}}  t^{2n}~,  \quad
\left[ \log a \right]_n =  \frac{1}{n} [1,0, \ldots,0]_{x^n}  t^n~.
\eea
Therefore,
\bea
\widehat{F}_{III} (a) &=& \exp \Bigg(- \sum_{n=1}^\infty \frac{1}{2n} [1,0, \ldots,0]_{x^{2n}}  t^{2n}  + \sum_{n=1}^\infty \frac{1}{2n} [1,0, \ldots,0]^2_{x^n}  t^{2n} \Bigg) \nn \\
&=& \exp \Bigg( \sum_{n=1}^\infty \frac{1}{n}  [0,1,0, \ldots,0]_{x^n}  t^{2n} \Bigg) \nn \\
&=& \PE \big[ [0,1,0, \ldots,0]_x t^2 \big]~,
\eea
where $[0,1,0, \ldots,0]$ is the rank two antisymmetric representation of $SU(2N_f)$.

\paragraph{The Hilbert series.} The Hilbert series is then given by
\bea
g_{N_f, C_n} (t,x) &=& G(a)^n \widehat{F}_{III} (a) \det (\BU + K^C_n) \nn \\
&=&  \det (\BU + K^C_n) \PE \big[ [0,1,0, \ldots,0]_x t^2 \big]~. \label{unrefBOBEHSDnref}
\eea 

\subsubsection*{Some examples}
Below we give certain explicit examples.

\paragraph{The case of $N_f \leq n$.}  In this case $K^C_n (i,j ) =0$ for all $i, j$.  Therefore, the Hilbert series is
\bea
g_{N_f \leq n} (t, x) &=&  \PE \big[ [0,1,0, \ldots,0]_x t^2 \big] \nn \\
&=& \sum_{m_2, m_4, \ldots, m_{2n} \geq 0} [0, m_2, 0, m_4, 0, \ldots, m_{2N_f-2},0] t^{\sum_{j=1}^n 2j m_{2j}}~.
\eea

\paragraph{The case of $N_f = n+1$.}  In this case 
\bea
K^C_n (i,j ) = \left\{ 
  \begin{array}{l l}
    -t^{2N_f} & \quad \text{if $i = j =n$}\\
    0 & \quad \text{otherwise}~.
  \end{array} \right. 
\eea
The Hilbert series is thus
\bea
g_{N_f = n+1, C_n} (t) &=& \left( 1-t^{2 N_f} \right) \PE \big[ [0,1,0, \ldots,0]_x t^2 \big] \nn \\
&=& \sum_{m_2, m_4, \ldots, m_{2n} \geq 0} [0, m_2, 0, m_4, 0, \ldots, m_{2n},0] t^{\sum_{j=1}^n 2j m_{2j}}~.
\eea

\paragraph{The case of $N_f = n+2$.}  The non-trivial block of the matrix $K^C_n$ is given by
{\small
\bea
\left(
\begin{array}{ccc}
K^C_n (n,n) \qquad & K^C_n (n,n+1) \qquad & -t^{4+2 n} \\
K^C_n (n,n+1)  & -t^{4+2 n} & 0 \\
 -t^{4+2 n} & 0 & 0
\end{array}
\right)~,
\eea}
where
\bea
K^C_n (n,n) &=& -C_{2n+2} \nn \\
&=& - \Big( [0,\ldots,0,1,0] t^{2n+2} - ( [1,0,\ldots,0,1]+1)  t^{2n+4} \nn \\
&& \qquad + [2,0, \ldots,0] t^{2n+6} \Big)~,  \\
 K^C_n (n,n+1) &=& -C_{2n+3} \nn \\
 &=& - \Big( -[0,\ldots,0,1] t^{2n+3} + [1,0, \ldots,0]t^{2n+5} \Big)~.
\eea
Hence we obtain
\bea
\det (\BU + K^C_n) &=& 1+K^C_n (n,n) -  \left[ K^C_n (n,n+1) \right]^2 - \left[1+ K^C_n (n,n) \right] t^{2n+4} \nn \\
&& - t^{4n+8} +t^{6n+12}~,
\eea
and the Hilbert series is therefore
\bea
g_{N_f = n+2, C_n} (t) &=& \sum_{m_2, m_4, \ldots, m_{2n} \geq 0} [0, m_2, 0, m_4, 0, \ldots, m_{2n},0,0,0] t^{\sum_{j=1}^n 2j m_{2j}}~.\eea

\paragraph{General formula.}  Note that the results in the above examples are in agreement with (3.10) of \cite{Hanany:2008kn}:
\bea
g_{N_f , C_n} (t,x) = \sum_{m_2, m_4, \ldots, m_{2n} \geq 0} [0, m_2, 0, m_4, 0, \ldots, m_{2n},0,0,\ldots,0] t^{\sum_{j=1}^n 2j m_{2j}}~.
\eea


\end{document}